\documentclass[]{interact}

\newtheorem{theorem}{Theorem}

\newtheorem{definition}{Definition}

\newtheorem{proposition}{Proposition}
\newtheorem{remark}{Remark}

\usepackage{subfigure}
\usepackage{amsmath}
\usepackage{amsfonts}
\usepackage{amssymb}
\usepackage{graphicx}
\usepackage{comment}
\usepackage{mathrsfs}
\usepackage{bm}
\usepackage{color}

\begin{document}

\title{Application of maximal monotone operator method for solving Hamilton-Jacobi-Bellman equation arising from optimal portfolio selection problem}

\author{
\name{Cyril Izuchukwu Udeani
and
Daniel \v{S}ev\v{c}ovi\v{c}\thanks{CONTACT D.~\v{S}ev\v{c}ovi\v{c}. Email: sevcovic@fmph.uniba.sk}
}
\affil{
Comenius University in Bratislava, Mlynsk\'a dolina, 84248 Bratislava, Slovakia}
}




\maketitle

\begin{abstract}
In this paper, we investigate a fully nonlinear evolutionary Hamilton-Jacobi-Bellman (HJB) parabolic equation utilizing the monotone operator technique. We consider the HJB equation arising from portfolio optimization selection, where the goal is to maximize the conditional expected value of the terminal utility of the portfolio. The fully nonlinear HJB equation is transformed into a quasilinear parabolic equation using the so-called Riccati transformation method. The transformed parabolic equation can be viewed as the porous media type of equation with source term. Under some assumptions, we obtain that the diffusion function to the quasilinear parabolic equation is globally Lipschitz continuous, which is a crucial requirement for solving the Cauchy problem. We employ Banach's fixed point theorem to obtain the existence and uniqueness of a solution to the general form of the transformed parabolic equation in a suitable Sobolev space in an abstract setting. Some financial applications of the proposed result are presented in one-dimensional space. 
\end{abstract}

\medskip
\medskip
\noindent AMS-MOS Classification: {35K55, 34E05, 70H20, 91B70, 90C15, 91B16}


\section{Introduction}

We investigate the existence and uniqueness of a solution $\varphi=\varphi(x,\tau)$ to the Cauchy problem for the nonlinear parabolic PDE
\begin{eqnarray}
\label{generalPDE}
&&\partial_{\tau}\varphi -\Delta \alpha(\tau, \varphi) = g_0(\tau, \varphi) + \nabla\cdot \bm{g}_1(\tau, \varphi),
    \\
&&	\varphi(\cdot, 0) =\varphi_{0},
\end{eqnarray}
where  $\tau\in(0,T), x\in\mathbb{R}^d, d\ge 1$. The diffusion function $\alpha=\alpha(x,\tau,\varphi)$ is assumed to be a globally Lipschitz continuous and strictly increasing function in the $\varphi$-variable. An example of such a Lipschitz continuous function $\alpha(x,\tau,\varphi)$ is the value function of the following parametric optimization problem:
\begin{equation}
\alpha(x,\tau,\varphi) = \min_{ {\bm{\theta}} \in \triangle} 
\left(
-\mu(x,t,{\bm{\theta}}) +  \frac{\varphi}{2}\sigma(x,t,{\bm{\theta}})^2\right), \quad \tau\in(0,T), x\in\mathbb{R}^d, \varphi>\varphi_{min}\,,
\label{eq_alpha_general}
\end{equation}
where $\mu,\sigma^2$ are given $C^1$ functions and $\triangle\subset\mathbb{R}^n$ is a compact decision set. Depending on the structure of the decision set $\triangle$, the function $\alpha$ is $C^{1,1}$ smooth if $\triangle$ is a convex set. But it can be only $C^{0,1}$ smooth if $\triangle$ is not connected.

Problems related to nonlinear parabolic equation arise in several mathematical models of applied sciences, such as chemical reactions, population dynamic, economics, and finance, have attracted great attentions. Recently, there are many results about existence, uniqueness, blowing-up, global existence, and other properties of parabolic equations, see for example Wu et al. \cite{Wu et al}. Some of the authors who studied parabolic equations used the method of upper and lower solution, see  Pao and Ruan \cite{PaoRuan}. Following a different approach, we utilize fixed point theorem, Fourier transform, and monotone operator technique with some shift/perturbation, to study the existence and uniqueness of solution to the Cauchy problem for the nonlinear parabolic equation in an abstract settings. 

The motivation for studying nonlinear parabolic equation of the form (\ref{generalPDE}) in one dimensional space (i.e., $d=1$) arises from dynamic stochastic programming. The fully nonlinear Hamilton-Jacobi-Bellman (HJB) equation describing optimal portfolio selection strategy is represented by the following fully nonlinear parabolic equation:
\begin{eqnarray}
&& \partial_t V + \max_{ {\bm{\theta}} \in \triangle} 
\left(
\mu(x,t,{\bm{\theta}})\, \partial_x V 
+ \frac{1}{2} \sigma(x,t,{\bm{\theta}})^2\, \partial_x^2 V \right) = 0\,,  
\label{eq_HJB}
\\
&& V(x,T)=u(x),  \label{init_eq_HJB}
\end{eqnarray}
where $x\in\mathbb{R}, t\in [0,T)$. A solution $V=V(x,t)$ to the parabolic equation (\ref{eq_HJB}) is subject to the terminal condition $V(x,T)=u(x)$. Following the papers by Kilianov\'a and {\v S}ev{\v c}ovi{\v c} \cite{KilianovaSevcovicANZIAM, KilianovaSevcovicKybernetika, KilianovaSevcovicJIAM}, the Hamilton-Jacobi-Bellman equation of the form (\ref{eq_HJB}) arises from  dynamic stochastic programming, where a goal is to maximize the conditional expected value of the terminal utility of the portfolio:
\begin{equation}
\max_{{\bm{\theta}}|_{[0,T)}} \mathbb{E}
\left[u(x_T^{\bm{\theta}})\, \big| \, x_0^{\bm{\theta}}=x_0 \right],
\label{maxproblem}
\end{equation}
on a finite time horizon $[0,T]$. Here, $u: \mathbb{R} \to \mathbb{R}$ is a given increasing terminal utility function and $x_0$ is a given initial state condition of the process $\{x_t^{\bm{\theta}}\}$ at $t=0$. The underlying stochastic process $\{x_t^{{\bm{\theta}}}\}$  with a drift $\mu(x,t,{\bm{\theta}})$ and volatility $\sigma(x,t,{\bm{\theta}})$ is assumed to satisfy the following It\^{o}'s  stochastic differential equation (SDE):
\begin{equation}
\label{process_x}
d x_t^{\bm{\theta}} = \mu(x_t^{\bm{\theta}}, t, {\bm{\theta}}_t) dt + \sigma(x_t^{\bm{\theta}}, t, {\bm{\theta}}_t) dW_t\,,
\end{equation}
where the control process $\{{\bm{\theta}}_t\}$ is adapted to the process $\{x_t\}$. Here, $\{W_t\}$ is the standard one-dimensional Wiener process. We assume the control parameter ${\bm{\theta}}$ belongs to a given compact subset $\triangle$ in $\mathbb{R}^n$. As an example, one can consider a compact convex simplex $\triangle\equiv\mathcal{S}^n = \{{\bm{\theta}} \in \mathbb{R}^n\  |\  {\bm{\theta}} \ge \mathbf{0}, \mathbf{1}^{T} {\bm{\theta}} = 1\} \subset \mathbb{R}^n$, where $\mathbf{1} = (1,\cdots,1)^{T} \in \mathbb{R}^n$. 

If we introduce the value function
\begin{equation}
V(x,t):= \sup_{  {\bm{\theta}}|_{[t,T)}} 
\mathbb{E}\left[u(x_T^{\bm{\theta}}) | x_t^{\bm{\theta}}=x \right].
\end{equation}
Then, following Bertsekas \cite{Bertsekas}, the value function $V=V(x,t)$ satisfies the fully nonlinear Hamilton-Jacobi-Bellman (HJB) parabolic equation (\ref{eq_HJB}) and $V(x,T):=u(x)$. 

Several attempts have been made for solving the HJB equation (\ref{eq_HJB}).  In this paper, we concentrate on the case when the utility function $u$ is increasing, as a consequence, $\partial_x V(x,\tau) >0$. The analysis of solutions to a fully nonlinear parabolic equation modeling the problem of optimal portfolio construction was investigated by Macov{\'a} and {\v{S}}ev{\v{c}}ovi{\v{c}} \cite{MS}. They showed how the problem of optimal stock to bond proportion in the management of a pension fund portfolio could be formulated in terms of the solutions to the HJB equation. Utility maximization problem for an investment-consumption portfolio when the current utility depends on the wealth process - regularity of solutions to the HJB equation was investigated by Federicol et al. \cite{Federico}. They defined a dual problem and treated it by means of dynamic programming, indicating that the viscosity solutions of the associated HJB equation belong to a class of smooth function. Ishimura and \v{S}ev\v{c}ovi\v{c} \cite{IshSev} constructed and analyzed solutions to the class of Hamilton-Jacobi-Bellman equation (\ref{eq_HJB}) with range bounds on the optimal response variable. They constructed monotone traveling wave solutions and identified parametric regions for which the traveling wave solutions have positive or negative wave speed. Abe and Ishimura \cite{AI} employed the Riccati transformation method for solving the full nonlinear HJB equations. The so-called Riccati transformation was later studied and generalized by Kilianov\'a and \v{S}ev\v{c}ovi\v{c} \cite{KilianovaSevcovicANZIAM}. They investigated solutions of a fully nonlinear HJB equation for a constrained dynamic stochastic optimal allocation problem. However, no attempt has been made in solving the fully nonlinear Hamilton-Jacobi-Bellman parabolic equation arising in portfolio optimization in a suitable Sobolev space using the monotone operator technique. The monotone operator method is essential because it does not only give constructive proof for existence theorems, but it also leads to various comparison results, which are effective tools for studying qualitative properties of solutions.

In this paper, inspired by the above studies, we investigate the existence and uniqueness of a solution to the Cauchy problem for the nonlinear parabolic PDE (\ref{generalPDE}) in suitable Sobolev spaces using the monotone operator approach. Employing the so-called Riccati transformation with some shift in the underlying operator, the HJB equation (\ref{eq_HJB}) can be transformed to the Cauchy problem for the nonlinear PDE (\ref{generalPDE}). We first show that the underlying abstract operator to the proposed Cauchy problem is strongly monotone in a suitable Sobolev space. Employing the monotonicity of the operator, Banach fixed point theorem, and Fourier transform approach, we obtain the existence and uniqueness of a solution to the Cauchy problem (\ref{generalPDE}) in an abstract setting.

The remainder of the paper is organized as follows. In order to motivate our study of equation (\ref{generalPDE}), we consider the fully nonlinear HJB partial differential equation  (\ref{eq_HJB}).  In Section 2, we present the existence and uniqueness result, in an abstract setting, of a solution to the Cauchy problem (\ref{generalPDE}) in Theorem \ref{th:alpha-existence}. The proof of the proposed theorem in a suitable Sobolev space is based on the monotone operator argument with the combination of Banach's fixed point and Fourier transforms techniques. Furthermore, we deduce smoothness properties of solutions. In Section 3, we introduce the so-called Riccati transformation of the HJB parabolic equation and illustrate its application to optimal portfolio selection problem. We investigate the relationship between the fully nonlinear HJB equation and transformed quasilinear Cauchy equation. The transformed function can be interpreted as the coefficient of the relative risk aversion of an investor. We also present the properties of the value function as a diffusion function. The aim is to show that the value function is globally Lipschitz continuous and strictly increasing, which are crucial requirements for obtaining the existence and uniqueness of a solution to the transformed Cauchy equation. The point-wise a-priori estimates of solutions, their existence and uniqueness are investigated. In Section 4, we present some numerical examples to illustrate the proposed result in the one-dimensional space. Finally, Section 5 contains the conclusion.

\section{Existence and uniqueness of a solution to the Cauchy problem}

We begin with the definition of the function spaces we will work with. Let $V\hookrightarrow H \hookrightarrow V^\prime$ be a Gelfand triple, where 
\[
H = L^2(\mathbb{R}^d) = \{ f:\mathbb{R}^d\to\mathbb{R}, \Vert f\Vert_{L^2}^2 = \int_{\mathbb{R}^d} |f(x)|^2 dx < \infty \}
\]
is a Hilbert space endowed with the inner product $(f,g)= \int_{\mathbb{R}^d} f(x) g(x) dx$. The Banach spaces $V, V^\prime$ are defined as follows: 
\[
\quad V = H^1(\mathbb{R}^d), \quad V^\prime = H^{-1}(\mathbb{R}^d),
\]
where the Sobolev spaces $H^s(\mathbb{R}^d)$ are defined by means of the Fourier transform
\[
\hat f(\xi) = \frac{1}{(2\pi)^{d/2}} \int_{\mathbb{R}^d} e^{-i x\cdot\xi} f(x) dx, \quad \xi = (\xi_1, \xi_2, ..., \xi_d)^T\in\mathbb{R}^d,
\]
\[
H^s(\mathbb{R}^d) = \{ f:\mathbb{R}^d\to\mathbb{R}, (1+|\xi|^2)^{s/2} \hat f(\xi) \in L^2(\mathbb{R}^d) \}, \; s\in\mathbb{R}
\]
endowed with the norm $\Vert f\Vert_{H^s}^2 = \int_{\mathbb{R}^d} (1+|\xi|^2)^s |\hat f(\xi)|^2 d\xi$, where $ |\xi| = (\xi_1^2+ \cdots + \xi_d^2)^{1/2}$

Let us introduce the linear operator $A : V\to V^\prime$ as follows:
\[
A \psi = \psi  - \Delta \psi, 
\]
Note that $A$ is a self-adjoint operator in the Hilbert space $H=L^2(\mathbb{R}^d)$ 
having the Fourier transform representation:
\[
\widehat{A\psi}(\xi) = (1+|\xi|^2)\hat\psi(\xi).
\]
Furthermore, the fractional powers of $A$ can be defined as follows: $\widehat{A^s\psi}(\xi) = (1+|\xi|^2)^s\hat\psi(\xi), \; s\in \mathbb{R}$. In particular, 
\[
\widehat{A^{\pm1/2}\psi}(\xi) = (1+|\xi|^2)^{\pm1/2}\hat\psi(\xi), \quad \]
and $A^{-1/2}$ is a self-adjoint operator in the Hilbert space $H=L^2(\mathbb{R}^d)$. Moreover, $A^{-1}= A^{-1/2} A^{-1/2}$.

In the sequel, we shall denote the duality pairing between the spaces $ V$ and $ V^\prime$ by $ \langle .,.\rangle$, i.e., the value of a functional $ F\in V^\prime$ at $u\in V$ is denoted by $ \langle F, u \rangle$. We have the following definitions.
\begin{definition}
\cite{Barbu}  An operator (in general nonlinear) $ B: V\to V^\prime $ is said to be
\begin{itemize}
    \item  [(i)] monotone if $$ \langle B(u) - B(v), u-v\rangle \geq 0, \; \forall \; u,v\in V,$$
    \item [(ii)] strongly monotone if there exists a constant $C>0$ such that
    $$ \langle B(u) - B(v), u-v\rangle \geq C\Vert u-v\Vert_V^2,  \; \forall \; u,v\in V,$$ 
    \item  [(iii)] hemicontinuous if for each $ u, v\in V$, the real-valued function $ t\mapsto B(u+tv)(v) $ is continuous.
\end{itemize}
\end{definition}

\begin{theorem}\cite{Showalter, Barbu}
\label{th:Showalter}
Let $V$ be a separable reflexive Banach space, dense and continuous in a Hilbert space $H$ which is identified with its dual, so $V \hookrightarrow H\hookrightarrow V^\prime$. Let $p\geq 2$ and set $\mathcal{V} = L^p ((0,T); V).$ Assume a family of operators $\mathcal{A}(\tau,.): V\to V^\prime, 0\leq \tau< T$, is given such that
\begin{itemize}
\item [(i)] for each $\varphi \in V $, the function $\mathcal{A}(.,\varphi): [0,T]\to V^\prime $ is measurable,
\item [(ii)] for a.e $\tau\in [0,T]$, the operator $\mathcal{A}(\tau,.): V\to V^\prime$ is monotone, hemicontinuous and bounded by
$\Vert \mathcal{A}(\tau,\varphi)\Vert \leq C ( \Vert \varphi\Vert^{p-1} + k(\tau)), \varphi\in V, 0\leq \tau< T,$ 
where $ k\in L^{p'}(0,T) $, 
\item [(iii)] and there exists  $\lambda >0$ such that 
$\langle \mathcal{A}(\tau,\varphi), \varphi\rangle  \geq \lambda\Vert \varphi\Vert^{p} - k(\tau), \varphi\in V, 0\leq \tau< T.$
\end{itemize}
Then for each $\hat f\in \mathcal{V^\prime}$ and $ \varphi_{0}\in H$, there exists a unique solution $\varphi\in\mathcal{V}$ of the Cauchy problem
\[
\partial_\tau\varphi(\tau) + \mathcal{A}(\tau,\varphi(\tau)) = \hat f(\tau) ~ \text{in}~ \mathcal{V^\prime},\ \  \varphi(0) =\varphi_{0}.
\]
\end{theorem}

\medskip
In what follows, we consider the spaces $\mathcal{V} = L^2 ((0,T);V)$, $\mathcal{H} = L^{2}((0,T);H) $  and  $\mathcal{V^\prime} = L^2 ((0,T);V^\prime)$, i.e., $p=2$. So we have the Gelfand triple $\mathcal{V}\hookrightarrow \mathcal{H} \hookrightarrow \mathcal{V^\prime}$, where $ \mathcal{H}$ is a Hilbert space endowed with the norm 
\[
\Vert \varphi\Vert^{2}_{\mathcal{H}} = \int_{0}^{T}\Vert \varphi(\tau)\Vert^{2}_{H}d\tau, \; \forall\varphi\in\mathcal{H}.
\]  

For a given value $\varphi_{min}$, we denote ${\mathcal D}= \mathbb{R}^d\times(0,T)\times (\varphi_{min},\infty)$. 

\begin{theorem}
\label{th:alpha-existence}
Assume that the above settings on $H$ and $V$ hold. Let $ g_0, g_{1j}:[0,T]\times H\to H, j=1,\cdots, n,$ be globally Lipschitz continuous functions. Suppose $\alpha\in C^{0,1}(\mathcal{D})$ is such that there exist constants $\omega, L, L_0>0$ such that $0 < \omega\le  \alpha^\prime_\varphi(x,\tau,\varphi)\le L$, $|\nabla_x \alpha(x,\tau,\varphi)|\le p(x,\tau) + L_0  |\varphi|$, $\alpha(x,\tau,0)=h(x,\tau)$ for  a.e. $(x,\tau,\varphi)\in\mathcal{D}$ and $p, h\in L^{\infty}((0,T);H)$.
Then for any $T>0 \; \text{and} \;\varphi_{0}\in H,$ there exists a unique solution $\varphi\in{\mathcal V}$ of the Cauchy problem 
\begin{equation}
	\partial_{\tau}\varphi + A\alpha(\cdot,\tau, \varphi) = g_0(\tau, \varphi) + \nabla \cdot \bm{g}_1(\tau, \varphi), \qquad 
	\varphi(0) =\varphi_{0}. 
	\label{equ:g_0,g_1}
\end{equation}

\end{theorem}

\noindent P r o o f: Recall that $ H = L^{2}(\mathbb{R}^d) $ and $ V = H^{1}(\mathbb{R}^d)$, its dual space being $ V^\prime = H^{-1}(\mathbb{R}^d)$.  Let the scalar products in $V$ and $ V^{'}$ be defined as follows:
\[
(f, g)_V =(A^{1/2}f, A^{1/2} g)_{H}=(Af, g)_{H}, \ 
(f, g)_{V^\prime} =(A^{-1/2}f, A^{-1/2} g)_{H}=(A^{-1}f, g)_{H},
\] respectively. Let us define the operator $\mathcal{A}(\tau, \cdot) : V\to V^\prime$ by 
\[
\langle\mathcal{A}(\tau, \varphi), \psi\rangle = ( A^{-1}A\alpha(\cdot,\tau, \varphi), \psi)_{H} =(\alpha(\cdot, \tau, \varphi), \psi)_{H}.
\]

Under the assumption made on the function $\alpha$ we can conclude that the mapping $\varphi\mapsto\alpha(\cdot,\tau,\varphi)$ maps $V$ into $V$. Indeed, if $\varphi\in V$ and $\eta=\alpha(\cdot,\tau,\varphi)$ then $\eta(x) = \alpha(x,\tau,\varphi(x)) - \alpha(x,\tau,0) + \alpha(x,\tau,0)$ and so
\[
|\eta(x)| \le (\max_\varphi \alpha^\prime_\varphi(x,\tau,\varphi)) |\varphi(x)|  + |h(x,\tau)| 
\le L |\varphi(x)| + |h(x,\tau)|. 
\]
Thus, $\int_{\mathbb{R}^d} |\eta(x)|^2 dx \le  2 \int_{\mathbb{R}^d} L^2 |\varphi(x)|^2  +  |h(x,\tau)|^2 dx \le 2L^2\Vert\varphi\Vert^2_H + 2 \Vert h(\cdot,\tau)\Vert^2_H$. Since $\nabla\eta(x) = \nabla_x\alpha(x,\tau, \varphi(x)) + \alpha^\prime_\varphi(x,\tau, \varphi(x))\nabla\varphi(x)$, we have
 \begin{equation*}
\begin{split}
 \Vert \eta\Vert_V^2=&\int_{\mathbb{R}^d} |\eta(x)|^2 + |\nabla\eta(x)|^2 dx 
\\ \leq & 2 \int_{\mathbb{R}^d} L^2 |\varphi(x)|^2  +  |h(x,\tau)|^2 dx + 2 \int_{\mathbb{R}^d} |p(x,\tau)|^2 + L_0^2|\varphi(x)|^2 dx  +  2\int_{\mathbb{R}^d} L^2|\nabla \varphi(x)|^2 dx \\ \leq & 2(L^2 \Vert \varphi \Vert_V^2 + \Vert h(\cdot, \tau)\Vert_H^2 + \Vert p(\cdot ,\tau)\Vert_H^2 + L_0^2 \Vert \varphi\Vert_H^2)<\infty,
\end{split}
\end{equation*}
because $p,h\in L^{\infty}((0,T);H)$. Consequently,  $\eta\in V$, as claimed.  

Next, we show that the operator $\mathcal{A}$ is monotone in the space $ V^\prime$. According to (\ref{lipschitz}) we have $(\alpha(x,\tau,\varphi_{1}) - \alpha(x,\tau,\varphi_{2}))(\varphi_{1} - \varphi_{2})\geq \omega (\varphi_{1} - \varphi_{2})^2$, for any $\varphi_1, \varphi_2\ge\varphi_{min}, x\in\mathbb{R}, \tau\in[0,T]$. 
 \begin{equation*}
\begin{split}
\langle\mathcal{A}(\tau,\varphi_{1}) -\mathcal{A}(\tau,\varphi_{2}), \varphi_{1} -\varphi_{2}\rangle & = (\alpha (\cdot,\tau,\varphi_{1}) - \alpha (\cdot,\tau, \varphi_{2}), \varphi_{1} - \varphi_{2}) 
\\ & = \int_{\mathbb{R}^d}(\alpha (x,\tau,\varphi_{1}(x)) - \alpha (x,\tau, \varphi_{2}(x)))( \varphi_{1}(x) - \varphi_{2}(x)) dx 
\\ & \geq \int_{\mathbb{R}^d}\omega|\varphi_{1}(x) - \varphi_{2}(x)|^{2} dx  = \omega \Vert \varphi_{1} - \varphi_{2}\Vert^{2}_{H}.
\end{split}
\end{equation*}
This implies that the operator $ \mathcal{A}(\tau,\cdot)$ is strongly monotone. 

For a given $\tilde\varphi\in {\mathcal H}$, we have $\hat{f}\in\mathcal{V}^\prime$, where $\hat{f}(\tau) =  g_0(\tau,\tilde \varphi(\cdot, \tau)) + \nabla\cdot \bm{g}_1(\tau, \tilde\varphi(\cdot, \tau))$, because $g_0, g_{1j}: [0,T]\times H\to H $ are globally Lipschitz continuous, $H\hookrightarrow V^\prime$, and the operator $\nabla$ maps $H$ into $V^\prime$. The hemicontinuity, boundedness, and coercivity of the operator $\mathcal{A}$ follows from the assumption that $\alpha $ is globally Lipschitz continuous and strictly increasing.

Applying Theorem~\ref{th:Showalter} we deduce the existence of  a unique solution $ \varphi \in {\mathcal V} $ such that 
\begin{equation}
	\partial_{\tau}\varphi + \mathcal{A}(\tau,\varphi) = \hat{f}(\tau), \qquad \varphi_{0}\in H, 
	\label{eq:hatf}
\end{equation}
where $ \mathcal{A}(\tau,\varphi) = A\alpha(\cdot,\tau,\varphi)$. Next, we multiply (\ref{eq:hatf}) by $A^{-1}$ to obtain
\begin{equation}
\partial_{\tau}A^{-1}\varphi + \alpha(\cdot,\tau, \varphi) = f,
\label{eq:f}
\end{equation}
where $f=f(\tau,\tilde\varphi) = A^{-1}\hat{f}(\tau).$  For $\tau\in[0,T]$, we denote $\tilde{f}(\tilde\varphi) = A^{-1/2} \hat{f}(\tau) =  A^{-1/2} g_0(\tau,\tilde\varphi) + A^{-1/2}\sum_{j=1}^d \partial_{x_j}g_{1j}(\tau,\tilde\varphi)$. For the Fourier transform of $\tilde f$, we have
\[
\widehat{\tilde{f}(\tilde\varphi)}(\xi) = \frac{1}{(1+ |\xi|^2)^{1/2}}\widehat{g_0(\tau, \tilde\varphi})(\xi) + \sum_{j=1}^d \frac{(-i\xi_j)}{(1+ |\xi|^2)^{1/2}}\widehat{g_{1j}(\tau, \tilde\varphi})(\xi).
\]
Let $\beta>0$ be the Lipschitz constant of the mappings $g_0, g_{1j}, j=1, \cdots, d$. Using Parseval's identity and Lipschitz continuity of $g_0, g_{1j}$ in $H$, we obtain, for $\tilde\varphi_1, \tilde\varphi_2\in\mathcal{H}$, 
\begin{equation*}
\begin{split}
\Vert\tilde{f}(\tilde\varphi_{1}) - \tilde{f}(\tilde\varphi_{2})\Vert^{2}_{H} & = \Vert\widehat{\tilde{f}(\tilde\varphi_{1})} - \widehat{\tilde{f}(\tilde\varphi_{2})}\Vert^{2}_{H} = \int_{\mathbb{R}^d} |\widehat{\tilde{f}(\tilde\varphi_{1})}(\xi)- \widehat{\tilde{f}(\tilde\varphi_{2})}(\xi)|^{2}d\xi \\& 
\le
2 \int_{\mathbb{R}^d} \frac{1}{1+|\xi|^{2}}|\widehat{g_0(\tau,\tilde\varphi_{1})}(\xi) - \widehat{g_0(\tau,\tilde\varphi_{2})}(\xi)|^{2} 
\\ &
\quad + \sum_{j=1}^d \frac{|\xi|^2}{1+|\xi|^{2}}|\widehat{g_{1j}(\tau,\tilde\varphi_{1})}(\xi) - \widehat{g_{1j}(\tau,\tilde\varphi_{2})}(\xi)|^{2}
d\xi \\ & 
\leq 2\Vert \widehat{g_0(\tau,\tilde\varphi_{1})} - \widehat{g_0(\tau,\varphi_{2})}\Vert^{2}_{H}  
+ 2\sum_{j=1}^d \Vert \widehat{g_{1j}(\tau,\tilde\varphi_{1})} - \widehat{g_{1j}(\tau,\varphi_{2})}\Vert^{2}_{H}
\\ &
= 2 \Vert g_0(\tau,\tilde\varphi_{1}) - g_0(\tau,\tilde\varphi_{2})\Vert^{2}_{H}
+ 2 \sum_{j=1}^d \Vert g_{1j}(\tau,\tilde\varphi_{1}) - g_{1j}(\tau,\tilde\varphi_{2})\Vert^{2}_{H}  
\\ &
\leq \tilde \beta^{2} \Vert \tilde\varphi_{1} - \tilde\varphi_{2}\Vert^{2}_{H},
\end{split}
\end{equation*} where $\tilde \beta^2= 2(1+d)\beta^2$. Hence, we obtain 
\begin{equation}
\Vert \tilde{f}(\tilde\varphi_{1}) - \tilde{f}(\tilde\varphi_{2})\Vert_{H} \leq \tilde\beta \Vert\tilde\varphi_{1} - \tilde\varphi_{2}\Vert_{H}.
	\label{eq:liphatf}
\end{equation}
\noindent 
Suppose $\varphi_{1}, \varphi_{2} \in \mathcal{H} $ are such that $\varphi_{1}=F(\tilde{\varphi_{1}})$ and $\varphi_{2} = F(\tilde{\varphi_{2}}).$ Here, the map  
$ F : \mathcal{H} \to \mathcal{H} $ is defined by $ \varphi = F(\tilde{\varphi}), $ where $ \varphi$ is a solution to the Cauchy problem
\begin{equation*}
\partial_{\tau}A^{-1}\varphi + \alpha(\cdot,\tau,\varphi) = f(\tau,\tilde{\varphi}), \qquad \varphi(0) =\varphi_{0}. 
\end{equation*}
Letting $ \varphi = \varphi_{1} -\varphi_{2} = F(\tilde{\varphi_{1}}) - F(\tilde{\varphi_{2}}),$ we obtain
\begin{equation}
\partial_{\tau}A^{-1}(\varphi_{1} - \varphi_{2}) + \alpha(\cdot,\tau,\varphi_{1}) -\alpha(\cdot,\tau,\varphi_{2}) = f(\tilde{\varphi}_{1}) -f(\tilde{\varphi}_{2}).
	\label{eq:dsol}
\end{equation}
Next, we multiply (\ref{eq:dsol})  by $ \varphi_{1} -\varphi_{2}$ and take the scalar product in the space $ H$ to obtain
\begin{eqnarray}
 (\partial_{\tau}A^{-1}(\varphi_{1} - \varphi_{2}), \varphi_{1} - \varphi_{2}) &+& (\alpha(\cdot,\tau,\varphi_{1}) -\alpha(\cdot,\tau,\varphi_{2}), \varphi_{1} - \varphi_{2}) \nonumber 
\\
&=& ( f(\tau, \tilde\varphi_{1}) -f(\tau, \tilde{\varphi}_{2}), \varphi_{1} - \varphi_{2}).
\label{eq:scph}
\end{eqnarray}
Using (\ref{eq:liphatf}) and the fact that $ A^{-1/2}$ is self-adjoint in $H$, then (\ref{eq:scph}) gives

\begin{equation*}
\begin{split}
\frac{1}{2}\frac{d}{d\tau} & \Vert A^{-1/2}(\varphi_{1}-\varphi_{2})\Vert^{2}_{H} + \omega\Vert\varphi_{1}- \varphi_{2}\Vert^{2}_{H}
\\
& \leq \langle f(\tau, \tilde{\varphi}_{1})-f(\tau, \tilde{\varphi}_{2}), \varphi_{1}-\varphi_{2}\rangle  = \langle A^{1/2}(f(\tau, \tilde\varphi_{1}) - f(\tau, \tilde\varphi_{2})), A^{-1/2}(\varphi_{1} -\varphi_{2})\rangle \\& \leq \Vert A^{1/2}(f(\tau, \tilde\varphi_{1}) - f(\tau, \tilde\varphi_{2}))\Vert_{H} \Vert \varphi_{1} -\varphi_{2}\Vert_{V^\prime}  = \Vert \tilde{f}(\tilde\varphi_{1}) - \tilde{f}(\tilde\varphi_{2})\Vert_{H} \Vert \varphi_{1} -\varphi_{2}\Vert_{V^\prime} \\& \leq \tilde \beta \Vert\tilde{\varphi_{1}} - \tilde{\varphi_{2}}\Vert_{H} \Vert \varphi_{1} -\varphi_{2}\Vert_{V^\prime}.  
\end{split}
\end{equation*}

\noindent This implies 
\[
\frac{1}{2}\frac{d}{d\tau}\Vert \varphi_{1}-\varphi_{2}\Vert^{2}_{V^\prime} + \omega\Vert\varphi_{1}- \varphi_{2}\Vert^{2}_{H} \leq \tilde\beta \Vert\tilde{\varphi_{1}} - \tilde{\varphi_{2}}\Vert_{H} \Vert \varphi_{1} -\varphi_{2}\Vert_{V^\prime}.
\]
Then, integrating on a small time interval $ [0,T]$ from $ 0$ to $t$ and noting that $ \varphi_{1}(0) = \varphi_{2}(0) =\varphi_{0}$, we obtain

\begin{equation*}
\begin{split}
\frac{1}{2}\Vert \varphi_{1}(\tau)-\varphi_{2}(\tau)\Vert^{2}_{V^\prime} & + \omega \int_{0}^{\tau} \Vert\varphi_{1}(s)- \varphi_{2}(s)\Vert^{2}_{H}ds 
\\ & 
\leq \tilde\beta \int_{0}^{\tau}\Vert\tilde{\varphi_{1}}(s) - \tilde{\varphi_{2}}(s)\Vert_{H} \Vert \varphi_{1}(s) -\varphi_{2}(s)\Vert_{V^\prime}ds 
\\ & 
\leq \tilde\beta \max\limits_{\tau\in[0,T]} ~ \Vert\varphi_{1}(\tau) - \varphi_{2}(\tau)\Vert_{V^\prime}\int_{0}^{T}\Vert \tilde{\varphi_{1}}(\tau) - \tilde{\varphi_{2}}(\tau)\Vert_{H}d\tau.
\end{split}
\end{equation*}
Taking the maximum over $ \tau \in [0,T]$ and using the fact that for any $ a,b\in\mathbb{R}, ~ ab \leq \frac{1}{2} a^{2} + \frac{1}{2}b^{2}$, we obtain 
\begin{equation*}
\begin{split}
\frac{1}{2}(\max\limits_{\tau\in[0,T]} ~ &\Vert\varphi_{1}(\tau) - \varphi_{2}(\tau)\Vert_{V^\prime})^{2}  + \omega \int_{0}^{T} \Vert\varphi_{1}(\tau)- \varphi_{2}(\tau)\Vert^{2}_{H}d\tau \\& \leq \tilde\beta \max\limits_{\tau\in[0,T]} ~ \Vert\varphi_{1}(\tau) - \varphi_{2}(\tau)\Vert_{V^\prime}  \int_{0}^{T}\Vert \tilde{\varphi_{1}}(\tau) - \tilde{\varphi_{2}}(\tau)\Vert_{H}d\tau
\\& 
\leq \frac{1}{2}(\max\limits_{\tau\in[0,T]} ~ \Vert\varphi_{1}(\tau) - \varphi_{2}(\tau)\Vert_{V^\prime})^{2}  + \frac{\tilde\beta^{2}}{2} (\int_{0}^{T}\Vert \tilde{\varphi_{1}}(\tau) - \tilde{\varphi_{2}}(\tau)\Vert_{H}d\tau)^{2}.
\end{split}
\end{equation*}
Using the Cauchy-Schwartz inequality, we obtain 
$\omega \int_{0}^{T} \Vert\varphi_{1}(\tau)- \varphi_{2}(\tau)\Vert^{2}_{H}d\tau \leq \frac{\tilde\beta^{2}}{2} \int_{0}^{T}d\tau \int_{0}^{T}\Vert \tilde{\varphi_{1}}(\tau) - \tilde{\varphi_{2}}(\tau)\Vert^{2}_{H}d\tau = \frac{\tilde\beta^{2} T}{2}\int_{0}^{T}\Vert \tilde{\varphi_{1}}(\tau) - \tilde{\varphi_{2}}(\tau)\Vert^{2}_{H}d\tau.
$
This implies that 
\[
\Vert F(\tilde{\varphi_{1}}) - F(\tilde{\varphi_{2}})\Vert^{2}_{\mathcal{H}} \leq \frac{\tilde\beta^{2} T}{2\omega}\Vert \tilde{\varphi_{1}} -\tilde{\varphi_{2}}\Vert^{2}_{\mathcal{H}}. 
\]
Thus, for $T$ sufficiently small such that $\frac{\tilde\beta^{2} T}{2\omega} <1$, the operator $F$ is a contraction on the space $ \mathcal{H}$; therefore by the  Banach fixed point theorem, $F$ has a unique fixed point in $ \mathcal{H}$. It is worth noting that $\tilde\beta$ and $ \omega$ are given such that they are independent of $T$. If $T>0$ is arbitrary, then we can apply a simple continuation argument. Indeed, if the solution exists in $(0,T_0)$ interval with $\frac{\tilde \beta^{2} T_0}{2 \omega} <1$, then starting from the initial condition $\varphi_0=\varphi(T_0/2)$ we can continue the solution $\varphi$ from the interval $(0,T_0)$ over the interval $(0, T_0)\cup (T_0/2, T_0/2 +T_0) \equiv (0, 3T_0/2)$. Continuing in this manner, we obtain the existence and uniqueness of a solution $\varphi\in{\mathcal H}$ defined on the time interval $(0,T)$. 

Finally, the solution belongs to the space $\mathcal V$ because the right-hand side, i.e., the function  $\hat f(\tau)=g_0(\tau, \varphi(\cdot, \tau)) + \nabla\cdot\bm{g}_1(\tau, \varphi(\cdot, \tau))$ belongs to ${\mathcal V}'$. Applying Theorem~\ref{th:Showalter} we conclude  $\varphi\in{\mathcal V}$, as claimed.
\hfill
$\diamondsuit$

\medskip

The following result shows that the unique solution is absolutely continuous and satisfies the a-priori energy estimates. Under assumption of the previous theorem we have $\alpha(\cdot, 0), g_0(\cdot, 0), g_{1j}(\cdot, 0) \in \mathcal{H}$. Here, the space $\mathcal{X} = L^{\infty} ((0,T);V^\prime)$ is endowed with the norm 
\[
\Vert \varphi\Vert^{2}_{\mathcal{X}} = \sup_{\tau\in[0,T]}\Vert \varphi(\tau)\Vert^{2}_{V^\prime}, \; \forall\varphi\in\mathcal{X}.
\] 

\begin{theorem}
\label{cor:abs-continuous}
Suppose that the functions $\alpha, g_0, g_{1j}$ fulfills the assumptions of Theorem~\ref{th:alpha-existence}. Then the unique solution $ \varphi \in\mathcal{V}$ to the Cauchy problem $ (\ref{eq:hatf})$ is absolutely continuous, i.e., $ \varphi\in C([0,T];H)$. Moreover, there exist a constant $\tilde{C}>0$, such that the unique solution satisfies the following inequality:
\begin{equation}
\Vert \varphi \Vert^2_{\mathcal{X}} 
+ \Vert \varphi\Vert^2_{\mathcal{H}} 
\leq \tilde{C} \bigl(
\Vert \varphi_0\Vert^2_{V^\prime}
+ \Vert \alpha(\cdot, 0)\Vert^2_{\mathcal{H}}
+ \Vert g_0(\cdot, 0)\Vert^2_{\mathcal{H}}
+ \sum_{j=1}^d \Vert g_{1j}(\cdot, 0)\Vert^2_{\mathcal{H}}
\bigr).
\label{prior_est}
\end{equation}
\end{theorem}

\begin{proof}
Since $ \hat{f}\in\mathcal{V}^\prime$, where $\hat{f} = g_0 + \nabla \cdot \bm{g}_{1}$ and $\mathcal{A}(\tau,\varphi)\in\mathcal{V}^\prime$, then $\partial_{\tau}\varphi\in\mathcal{V}^\prime $. Therefore, for each $\varphi_0 \in H$, we have $\varphi\in W$ where $W$ is the Banach space  $W=\{\varphi,  \varphi\in\mathcal{V}, \partial_\tau\varphi \in \mathcal{V}'\}$. According to \cite[Proposition \;1.2]{Showalter}, we have $W \hookrightarrow C([0,T];H)$.
Hence, the unique solution $ \varphi$ to the Cauchy problem belongs to the space $ C([0,T];H)$, as claimed.

\medskip

Next, we show that the unique solution satisfies a-priori energy estimate (\ref{prior_est}). Let $\varphi$ be a unique solution to the Cauchy problem (\ref{equ:g_0,g_1}). Multiply (\ref{eq:f}) by $\varphi$ and take the scalar product in $H$ to obtain
\begin{equation}
(\partial_{\tau}A^{-1}\varphi, \varphi)_{H} + (\alpha(\cdot,\tau,\varphi), \varphi)_{H} =(A^{-1}g_0(\tau, \varphi) + A^{-1} \nabla\cdot \bm{g}_1(\tau, \varphi), \varphi).
\label{equ:scalarH}
\end{equation}
Using the Lipschitz continuity of $g_0, \bm{g}_1$ and strong monotonicity of $\alpha$,  we obtain
\begin{eqnarray*}
\frac{1}{2}\frac{d}{d\tau} \Vert \varphi\Vert^2_{V'} + \omega \Vert \varphi\Vert^2_{H} 
&=& 
(\partial_{\tau}A^{-1}\varphi, \varphi) +  \omega \Vert \varphi\Vert^2_{H} 
\\
&\le&
(\partial_{\tau}A^{-1}\varphi, \varphi) + (\alpha(\cdot,\varphi) - \alpha(\cdot, 0), \varphi)  
\\
&=& 
(A^{-1} ( g_0(\cdot, \varphi) + \nabla\cdot \bm{g}_1(\tau, \varphi) ) - \alpha(\cdot, 0), \varphi)
\\
&=& 
(A^{-1} ( g_0(\cdot, \varphi) - g_0(\cdot, 0)  + \nabla\cdot \bm{g}_1(\cdot, \varphi) - \nabla\cdot \bm{g}_1(\cdot, 0) ), \varphi)
\\
&& + (A^{-1} ( g_0(\cdot, 0)  + \nabla\cdot \bm{g}_1(\cdot, 0) ), \varphi) - (\alpha(\cdot, 0), \varphi)
\\
&=& (A^{-1/2} ( g_0(\cdot, \varphi) - g_0(\cdot, 0)  + \nabla\cdot \bm{g}_1(\cdot, \varphi) - \nabla\cdot \bm{g}_1(\cdot, 0) ), A^{-1/2} \varphi)
\\
&& + (A^{-1/2} ( g_0(\cdot, 0)  + \nabla\cdot \bm{g}_1(\cdot, 0) ), A^{-1/2} \varphi) - (\alpha(\cdot, 0), \varphi)
\\
&\le& 
\beta (1+d) \Vert \varphi\Vert_H \Vert\varphi\Vert_{V'}
+ \Vert A^{-1/2} ( g_0(\cdot, 0)  + \nabla\cdot \bm{g}_1(\cdot, 0) )\Vert_H \Vert\varphi\Vert_{V'} 
\\
&& + \Vert\alpha(\cdot,0)\Vert_{H}\Vert\varphi\Vert_{H}
\\
&\le& 
\frac{\omega}{4} \Vert \varphi\Vert^2 _H + \frac{\beta^2 (1+d)^2}{\omega} \Vert\varphi\Vert^2_{V'}
+ \frac{1}{2}\Vert A^{-1/2} ( g_0(\cdot, 0)  + \nabla\cdot \bm{g}_1(\cdot, 0) )\Vert^2_H  
\\&& 
+ \frac{1}{2}\Vert\varphi\Vert^2_{V'}  
+ \frac{1}{\omega}\Vert\alpha(\cdot,0)\Vert^2_{H} + \frac{\omega}{4}\Vert\varphi\Vert^2_{H}.
\end{eqnarray*}
Hence, there exist constants $C_0, C_1>0$ such that 
\begin{eqnarray*}
\frac{d}{d\tau} \Vert \varphi\Vert^2_{V'} + \omega \Vert \varphi\Vert^2_{H} 
&\le& 
 C_1 \Vert\varphi\Vert^2_{V'} + C_0 \bigl( \Vert  g_0(\cdot, 0)\Vert^2_H + \sum_{j=1}^d  \Vert  g_{1j}(\cdot, 0)\Vert^2_H  +  \Vert \alpha(\cdot, 0)\Vert^2_H \bigr).
\end{eqnarray*}
Solving the differential inequality $y'(\tau) \le C_1 y(\tau) + r(\tau) $, where $y(\tau) = \Vert\varphi(\cdot,\tau)\Vert^2_{V'}$ and $r(\tau) = C_0\bigl( \Vert  g_0(\cdot, \tau, 0)\Vert^2_H + \sum_{j=1}^d  \Vert  g_{1j}(\cdot, \tau, 0)\Vert^2_H +  \Vert \alpha(\cdot,\tau, 0)\Vert^2_H \bigr)$, yields 
\[
y(\tau) \le e^{C_1 T} \bigl(y(0) + \int_0^T r(s)ds \bigr),
\]
and the proof of the Theorem follows. 
\end{proof}

\section{The Riccati transformation of the HJB equation and application to optimal portfolio selection problem}
\label{sec:HJB}

\subsection{The Riccati transformation}

In this section, we present how the HJB equation (\ref{eq_HJB}) can be transformed into a quasilinear PDE, which is equivalent to the Cauchy problem for the nonlinear parabolic equation (\ref{generalPDE}).

Following the methodology introduced by Abe and Ishimura \cite{AI}, Ishimura and \v{S}ev\v{c}ovi\v{c} \cite{IshSev}, \v{S}ev\v{c}ovi\v{c} and Macov\'a \cite{MS}, and Kilianov\'a and \v{S}ev\v{c}ovi\v{c} \cite{KilianovaSevcovicANZIAM}, the Riccati transformation $\varphi$ of the value function $V$ can be introduced as follows:

\begin{equation}
\varphi(x,\tau) = - \frac{\partial_x^2 V(x,t)}{\partial_x V(x,t)}, \quad\hbox{where}\ \ \tau=T-t.
\label{eq_varphi}
\end{equation}

Suppose for a moment that the value function $V(x,t)$ is increasing in the $x$-variable. This is a natural assumption in the case when the terminal utility function $u(x)$ is increasing itself. Then the HJB equation (\ref{eq_HJB}) can be rewritten as follows:
\begin{equation}
\partial_t V - \alpha(\cdot,\varphi) \partial_x V = 0, \qquad V(\cdot ,T)=u(\cdot),
\label{eq_HJBtransf}
\end{equation}
where $\alpha(x,\tau,\varphi)$ is the value function of the following parametric optimization problem:
\begin{equation}
\alpha(x,\tau,\varphi) = \min_{ {\bm{\theta}} \in \triangle} 
\left(
-\mu(x,t,{\bm{\theta}}) +  \frac{\varphi}{2}\sigma(x,t,{\bm{\theta}})^2\right), \quad \tau=T-t\,.
\label{eq_alpha_def}
\end{equation}

\begin{remark}
The optimization problem (\ref{eq_alpha_def}) is related to the classical  Markowitz model on optimal portfolio selection problem formulated as maximization of the mean return $\mu({\bm{\theta}})\equiv {\bm{\mu}}^T {\bm{\theta}} $ under the volatility constraint $\frac12\sigma({\bm{\theta}})^2\equiv\frac12 {\bm{\theta}}^T{\bm{\Sigma}} {\bm{\theta}} \le  \frac12\sigma^2_0$, i.e.:
\begin{eqnarray*}
&&\max_{{\bm{\theta}}\in\triangle} {\bm{\mu}}^T {\bm{\theta}},
\quad s.t.\ \ \frac12 {\bm{\theta}}^T{\bm{\Sigma}} {\bm{\theta}} \le  \frac12\sigma^2_0,
\end{eqnarray*}
where the decision set is the simplex
$\triangle= \{{\bm{\theta}} \in \mathbb{R}^n\  |\  {\bm{\theta}} \ge \mathbf{0}, \mathbf{1}^{T} {\bm{\theta}} = 1\}$. Indeed, the Lagrange multiplier for the volatility constraint can be identified as the parameter $\varphi$ entering the parametric optimization problem  (\ref{eq_alpha_def}).
\end{remark}

In what follows, we shall denote by $\partial_x\alpha$ the total differential of the function $\alpha(x,\tau,\varphi)$ where $\varphi=\varphi(x,\tau)$, that is
\[
\partial_x\alpha (x,\tau,\varphi) = \alpha^\prime_x(x,\tau,\varphi) + \alpha^\prime_\varphi(x,\tau,\varphi)\, \partial_x \varphi,
\]
where $\alpha^\prime_x$ and $\alpha^\prime_\varphi$ are partial derivatives of $\alpha$ with respect to variables $x$ and $\varphi$, respectively.

The relationship between the transformed function $\varphi$ and the value function $V$ is given by the result due to Kilianov\'a and \v{S}ev\v{c}ovi\v{c} \cite{KilianovaSevcovicKybernetika}. With regard to \cite[Theorem 4.2]{KilianovaSevcovicKybernetika}, an increasing value function $V(x,t)$ in the $x$-variable is a solution to the Hamilton-Jacobi-Bellman equation (\ref{eq_HJB}) if and only if the transformed function $\varphi(x,\tau) = -\partial_x^2 V(x,t)/\partial_x V(x,t),\,\ t=T-\tau$,  is a solution to the Cauchy problem for quasi-linear parabolic PDE:
\begin{eqnarray}
&&\partial_\tau \varphi - \partial^2_x \alpha(\cdot,\varphi) = - \partial_x \left( \alpha(\cdot,\varphi)\varphi\right),
\label{eq_PDEphi-cons}
\\
&&\varphi(x,0) = \varphi_0(x) \equiv -u''(x)/u'(x),\quad (x,\tau)\in\mathbb{R}\times(0,T). 
\label{init_PDEphi-cons}
\end{eqnarray}

It is worth noting that the Cauchy problem for the quasi-linear parabolic PDE (\ref{eq_PDEphi-cons}) is equivalent to the nonlinear parabolic equation (\ref{generalPDE}) in one-dimensional space. This is obtainable after some shift/perturbation in the main operator of the transformed equation (\ref{eq_PDEphi-cons}).

\subsection{Properties of the value function as a diffusion function}

This section investigates qualitative properties of the value function and sufficient conditions imposed on the decision set $\triangle$ and functions $\mu$ and $\sigma$ that guarantee higher smoothness of the value function $\alpha$. Let us denote by $C^{k,1}(\mathcal{D})$ the space consisting of all $k$-differentiable functions defined on the domain $\mathcal{D}\subset \mathbb{R}^{d+2}$, whose $k$-th derivative is globally Lipschitz continuous. The next  proposition shows (under certain assumptions) that the value function $\alpha$ belongs to $C^{0,1}(\mathcal{D})$, where $\mathcal{D}=\mathbb{R}^d\times (0,T)\times (\varphi_{min},\infty)$.

\begin{proposition}\label{smootheness0}
Let $\triangle\subset\mathbb{R}^n$ be a given compact decision set. Assume that the functions $\mu(x,t,{\bm{\theta}})$ and $\sigma(x,t,{\bm{\theta}})^2$ are globally Lipschitz continuous in $x\in\mathbb{R}^d, t\in[0,T]$ and ${\bm{\theta}}\in\triangle$ variables, and there exist positive constants $\omega, L>0$ such that $\omega\le \frac12\sigma(x,t,{\bm{\theta}})^2\le L$ for any $x\in\mathbb{R}^d, t\in[0,T]$, and ${\bm{\theta}}\in\triangle$. 

Then  $\alpha\in C^{0,1}(\mathcal{D})$. Moreover, the function $\alpha$ is strictly increasing, and
\begin{equation}
0<\omega\le \frac{\alpha(x,\tau, \varphi_2)-\alpha(x,\tau, \varphi_1)}{\varphi_2-\varphi_1} \le L, 
\quad \text{for any}\ (x,\tau,\varphi_i)\in\mathcal{D}, 
\label{lipschitz}
\end{equation}
i.e., $\omega\le\alpha^\prime_\varphi(x,\tau,\varphi)\le L$, and
\begin{equation}
|\nabla_x\alpha(x,\tau,\varphi)| \le  p(x,\tau) + L_{0} |\varphi|, 
\label{lipschitz-x}
\end{equation}
for a.e. $(x,\tau,\varphi)\in\mathcal{D}$, where $p(x,\tau) := \max_{ {\bm{\theta}} \in \triangle} |\nabla_x\mu(x,t,{\bm{\theta}})|$
and $ L_0 := \max_{{\bm{\theta}}\in\triangle, t\in[0,T], x\in\mathbb{R}^d} |\nabla_x \sigma^2(x, t, \theta)|$ where $t=T-\tau$.
\end{proposition}

\begin{proof}
Let us define $\alpha^{\bm{\theta}}(x,\tau,\varphi) := 
-\mu(x,t,{\bm{\theta}}) +  \frac{\varphi}{2}\sigma(x,t,{\bm{\theta}})^2$, where $t=T-\tau$. Then
\[
\alpha(x,\tau,\varphi) = \min_{ {\bm{\theta}} \in \triangle} \alpha^{\bm{\theta}}(x,\tau,\varphi) \,.
\]
For any given ${\bm{\theta}}\in \triangle$,  the function $\alpha^{\bm{\theta}}(x,\tau,\varphi)$ is globally Lipschitz continuous in all variables. The minimal function $\alpha$ is therefore globally Lipschitz continuous as well. Moreover, the function $\alpha^{\bm{\theta}}(x,\tau,\varphi)$ satisfies the inequality (\ref{lipschitz}) for any ${\bm{\theta}}\in \triangle$, and so does the minimal function $\alpha$.

Next, we prove inequality (\ref{lipschitz-x}). Let $ x_1 ,x_2\in\mathbb{R}^d$ such that $x_2 = x_1 + h e^i$, where $e^i, i=1,\cdots, d,$ is the standard normal vector, i.e., $ e^i = (0,0,...,0,1,0,...,0)^T$. We have that 
\begin{eqnarray*}
&& \alpha^{\bm{\theta}}(x_1,\tau,\varphi)  - \alpha^{\bm{\theta}}(x_2,\tau,\varphi) = -(\mu(x_1,\tau,{\bm{\theta}})- \mu(x_2,\tau,{\bm{\theta}})) 
+ \frac{\varphi}{2} (\sigma(x_1,\tau,{\bm{\theta}})^2- \sigma(x_2,\tau,{\bm{\theta}})^2) 
\\
&& = \int_{0}^{h} (-\partial_{x_i}\mu(x_1+\xi e^i,\tau,{\bm{\theta}})) d\xi + \int_{0}^{h}  \frac{\varphi}{2} \partial_{x_i}\sigma^2(x_1+\xi e^i,\tau,{\bm{\theta}}) d\xi
\\
&&
\leq \int_{0}^{h} |\partial_{x_i}\mu(x_1+\xi e^i,\tau,{\bm{\theta}}))| d\xi + \int_{0}^{h}  \frac{|\varphi|}{2} |\partial_{x_i}\sigma^2(x_1+\xi e^i,\tau,{\bm{\theta}})| d\xi
\\
&& 
\le \max_{\bm{\theta}\in\triangle, 0\le\xi\le h} |\partial_{x_i}\mu(x_1+\xi e^i,\tau,{\bm{\theta}})| \, h  + \max_{{\bm{\theta}}\in\triangle, x\in\mathbb{R}^d} |\partial_{x_i} \sigma^2(x, \tau, \theta)|\frac{|\varphi|}{2}h.
\end{eqnarray*}
Hence,
\[
\alpha^{\bm{\theta}}(x_1,\tau,\varphi) \le  \alpha^{\bm{\theta}}(x_2,\tau,\varphi) +  \max_{{\bm{\theta}}\in\triangle, 0\le\xi\le h} |\partial_{x_i}\mu(x_1+\xi e^i,\tau,{\bm{\theta}})| \, h + \max_{{\bm{\theta}}\in\triangle, x\in\mathbb{R}^d} |\partial_{x_i} \sigma^2(x, \tau, \theta)||\varphi|h.
\]
We note that $x_2 - x_1 = h e^i$ so that $ |x_2 - x_1| = h$. Taking minimum over ${\bm{\theta}}\in\triangle$, we obtain 
\[
\alpha(x_1,\tau,\varphi) \le  \alpha(x_2,\tau,\varphi) +  \max_{\bm{\theta}\in\triangle,  0\le\xi\le h} |\partial_{x_i}\mu(x_1+ \xi e^i,\tau,{\bm{\theta}})| \, h + \max_{{\bm{\theta}}\in\triangle, x\in\mathbb{R}^d} |\partial_{x_i} \sigma^2(x, \tau, \theta)||\varphi| h .
\]
Exchanging the role of $x_1$ and $x_2$ and taking the limit as $x_2\to x_1$, i.e., $h\to 0$, we obtain inequality (\ref{lipschitz-x}), as stated.

\end{proof}

According to Proposition~\ref{smootheness0}, the value function $\alpha$ given in (\ref{eq_alpha_def}) fulfils the assumptions of Theorem~\ref{th:alpha-existence} provided that the functions 
\[
p(x,\tau)= \max_{{\bm{\theta}}\in\triangle} |\nabla_x\mu(x,\tau,{\bm{\theta}})|, \quad\text{and}\ 
h(x,\tau)= \alpha(x,\tau,0) = -\max_{{\bm{\theta}}\in\triangle} \mu(x,\tau,{\bm{\theta}})
\]
belong to the Banach space $L^\infty((0,T); H)$.\\

The next result was proved in \cite{KilianovaSevcovicJIAM}. It gives sufficient conditions imposed on the decision set $\triangle$ and functions $\mu$ and $\sigma$ guaranteeing higher smoothness of the value function $\alpha$. Its proof is based on the classical envelope theorem due to Milgrom and Segal \cite{milgrom_segal2002} and the result on Lipschitz continuity of the minimizer $\hat{\bm{\theta}}=\hat{\bm{\theta}}(x,\tau,\varphi)$ belonging to a convex compact set $\triangle$ due to Klatte \cite{Klatte}.

\begin{theorem}\cite[Theorem 1]{KilianovaSevcovicJIAM}
\label{smootheness1}
Suppose that $\triangle\subset\mathbb{R}^n$ is a convex compact set, and the functions $\mu(x,t,{\bm{\theta}})$ and $\sigma(x,t,{\bm{\theta}})^2$ are $C^{1,1}$ smooth such that the objective function 
$f(x,t,\varphi, {\bm{\theta}}):= - \mu(x,t,{\bm{\theta}}) + \frac{\varphi}{2} \sigma(x,t,{\bm{\theta}})^2$ is strictly convex in the variable ${\bm{\theta}}\in\triangle$ for any  $\varphi\in(\varphi_{min}, \infty)$, then the function $\alpha$ belongs to the space $C^{1,1}(\mathcal{D})$.
\end{theorem}

\subsection{Point-wise a-priori estimates of solutions, their existence and uniqueness}

In this section we present a-priori estimates on a solution $\varphi$ to the Cauchy problem (\ref{eq_PDEphi-cons}).
Throughout this section we will assume that the function $\mu$ is independent of time $t\in[0,T]$, and $\sigma$ is independent of $x\in\mathbb{R}$ and $t\in[0,T]$, i.e., 
\[
\mu=\mu(x,\bm{\theta}), \qquad \sigma = \sigma(\bm{\theta}).
\]
Then the value function $\alpha=\alpha(x,\varphi)$ is independent of the $\tau=T-t$ variable, as well. 

In what follows we will prove a-priori estimates for the transformed function $\psi=\psi(x,\tau)$ defined as $\psi(x,\tau)=\alpha(x,\varphi(x,\tau))$. Since $\alpha$ is strictly increasing function in the $\varphi$ variable, there exists an inverse function $\beta(x,\psi)$ such that $\alpha(x,\beta(x,\psi)) = \psi$. Straightforward calculations show that the function $\varphi(x,\tau)$ is a solution to (\ref{eq_PDEphi-cons}) if and only if the function $\psi(x,\tau)$ is a solution to the following linear parabolic PDE: 
\[
-\partial_\tau\psi + a(x,\tau) \partial_x^2 \psi + b(x,\tau) \partial_x \psi + c(x,\tau)\psi  = 0,
\]
where
\begin{eqnarray*}
a(x,\tau) &=& \alpha^\prime_\varphi(x,\varphi(x,\tau)), 
\quad
b(x,\tau) = - \alpha^\prime_\varphi(x,\varphi(x,\tau)) \varphi(x,\tau) - \alpha(x,\varphi(x,\tau)),
\\
c(x,\tau) &=& \alpha^\prime_x (x,\varphi(x,\tau)).
\end{eqnarray*}
Notice that $0<\omega \le a(x,\tau)\le L$. 
Suppose that the function $c(x,\tau)$ is bounded from above by a constant $\lambda\ge 0$. Then the function $\psi_\lambda(x,\tau) = \psi(x,\tau) e^{-\lambda\tau}$
is a solution to the linear PDE:
\[
{\mathcal L}[\psi_\lambda] = 0, \qquad \text{where}\ \  {\mathcal L}[\psi_\lambda] \equiv \partial_\tau\psi_\lambda - a(x,\tau) \partial_x^2 \psi_\lambda - b(x,\tau) \partial_x \psi_\lambda - c_\lambda(x,\tau)\psi_\lambda,
\]
where $c_\lambda(x,\tau)= c(x,\tau) -\lambda$ is nonpositive, i.e., $c_\lambda (x,\tau) \le0$ for all $x,\tau$.
Let $\underline{\psi}\le 0$ be a constant. Then ${\mathcal L}[\psi_\lambda - \underline{\psi}] =  c_\lambda(x,\tau)\underline{\psi}\ge 0$. Applying the maximum principle for parabolic equations on unbounded domains \cite[Theorem 3.4]{Meyer} due to Meyer and Needham, we obtain $\psi_\lambda(x,\tau) - \underline{\psi}\ge 0$ for all $x,\tau$ provided that $\psi_\lambda(x,0) - \underline{\psi}=\psi(x,0) - \underline{\psi}\ge 0$ for all $x$. That is, $\underline{\psi}$ is a subsolution. Similarly, if $\overline{\psi}\ge 0$ is a given constant, then ${\mathcal L}[\psi_\lambda - \overline{\psi}] =  c_\lambda(x,\tau)\overline{\psi}\le 0$ and $\psi_\lambda(x,\tau) - \overline{\psi}\le 0$ for all $x,\tau$ provided that $\psi(x,0) - \overline{\psi}\le 0$ for all $x$, i.e., $\overline{\psi} $ is a supersolution. In summary, we have the following implication:
\[
\underline{\psi}\le \psi(x,0) \le \overline{\psi} \quad \Longrightarrow \quad 
\underline{\psi} e ^{\lambda \tau}\le \psi(x,\tau) \le \overline{\psi} e ^{\lambda \tau} \quad\text{for all}\ x\in\mathbb{R}, \tau\in[0,T].
\]
In terms of the solution $\varphi$ to the Cauchy problem (\ref{eq_PDEphi-cons})--(\ref{init_PDEphi-cons}), we have the following a-priori estimate:
\begin{equation}
\underline{\psi} e ^{\lambda \tau}\le \alpha(x,\varphi(x,\tau)) \le \overline{\psi} e ^{\lambda \tau} \quad\text{for all}\ x\in\mathbb{R}, \tau\in[0,T], 
\label{pointwise}
\end{equation}
where 
\begin{equation}
\underline{\psi} = \min\{ 0, \inf_{x\in\mathbb{R}} \alpha(x,\varphi(x,0)) \}, \qquad 
\overline{\psi} = \max\{0, \sup_{x\in\mathbb{R}} \alpha(x,\varphi(x,0)) \}. 
\label{psiplusminus}
\end{equation}

Now, we are in a position to apply the general Theorem~\ref{th:alpha-existence} on existence and uniqueness of a solution. 

\begin{theorem}
Let the decision set $\triangle \subset \mathbb{R}^n$ be compact and the function $u:\mathbb{R}\to\mathbb{R}$ be an increasing utility function such that $\varphi_0(x) = -u''(x)/u'(x)$ belongs to the space $L^2(\mathbb{R})\cap L^\infty(\mathbb{R})$. Suppose that the drift $\mu(x, \bm{\theta})$ and volatility function $\sigma^2(\bm{\theta})>0$ are $C^1$ continuous in the $x$ and $\bm{\theta}$ variables, and the value function $\alpha(x,\varphi)$ given in (\ref{eq_alpha_def}) satisfies $p\in L^2(\mathbb{R})\cap L^\infty(\mathbb{R}), h\in L^\infty(\mathbb{R}), \; \text{and} \;\partial^2_x h \in L^2(\mathbb{R})$, where  
\begin{eqnarray*}
&&p(x)= \max_{{\bm{\theta}}\in\triangle} |\partial_x\mu(x,{\bm{\theta}})|, 
\quad 
h(x)=  -\max_{{\bm{\theta}}\in\triangle} \mu(x,{\bm{\theta}}).
\end{eqnarray*}
Then for any $T>0$ there exists a unique solution $\varphi$ of the Cauchy problem 

\begin{equation}
\label{finalequation}
\partial_\tau \varphi - \partial^2_x \alpha(\cdot,\varphi) = - \partial_x \left( \alpha(\cdot,\varphi)\varphi\right),
\quad \varphi(x,0) = \varphi_0(x),\quad (x,\tau)\in\mathbb{R}\times(0,T), 
\end{equation}
satisfying $\varphi\in C([0,T]; H)\cap L^2((0,T); V)\cap L^\infty((0,T)\times \mathbb{R})$.

\end{theorem}

\begin{proof}

Since $\sigma^2(\bm{\theta})>0$ and $\triangle$ is a compact set, there exist constants $0<\omega\le L$ such that $0<\omega\le \sigma^2(\bm{\theta})\le L$ for all $\bm{\theta}\in \triangle$. It follows from Proposition~\ref{smootheness0} that 
\begin{equation}
\omega|\varphi| \le |\alpha(x,\varphi) - \alpha(x,0)| \le L |\varphi|.
\label{monotonicity}
\end{equation}
Since $\varphi_0, h\in L^\infty(\mathbb{R})$ and $h(x) = \alpha(x,0)$, we obtain 
$M:=\sup_{x\in\mathbb{R}} |\alpha(x,\varphi_0(x))| <\infty$.

Let us define the shifted diffusion function by $\tilde\alpha(x,\varphi)=\alpha(x,\varphi) - \alpha(x,0)$. Notice that $\alpha(x,0)= \min_{{\bm{\theta}}\in\triangle} -\mu(x,{\bm{\theta}}) = h(x)$. Then equation (\ref{finalequation}) is equivalent to 
\[
\partial_\tau \varphi  + A \tilde \alpha(\cdot,\varphi) = 
\tilde \alpha(\cdot,\varphi)+ \partial^2_x h - \partial_x \left( \alpha(\cdot,\varphi)\varphi\right),
\]
where $A=I-\partial^2_x$.

Next, let $g_0(\varphi) = \tilde \alpha(\cdot,\varphi)+ \partial^2_x h$ and $g_1(\varphi) = -w(\alpha(\cdot,\varphi)) \varphi$. Here, $w:\mathbb{R}\to\mathbb{R}$ is a suitable cut-off function 
\[
w(\alpha) = \left\{ 
\begin{array}{ll}
\underline{\psi} e^{\lambda T}, & \text{if}\ \ \alpha  \le \underline{\psi} e^{\lambda T},  \\
\alpha, &  \text{if} \ \ \underline{\psi} e^{\lambda T} < \alpha < \overline{\psi} e^{\lambda T}, \\
\overline{\psi} e^{\lambda T}, & \text{if}\ \ \alpha \ge \overline{\psi} e^{\lambda T},  \\
\end{array}
\right.
\]
where $\overline{\psi}=M, \underline{\psi}=-M$.
Then the functions $g_0, g_1:H\to H$ are globally Lipschitz continuous. 

Notice that the diffusion function $\tilde\alpha$ fulfills assumptions of Theorem~\ref{th:alpha-existence} with $\tilde h(x)=\tilde\alpha(x,0)\equiv 0$. Now applying Theorems~\ref{th:alpha-existence} and \ref{cor:abs-continuous} we obtain the existence and uniqueness of a solution $\varphi\in C([0,T]; H)\cap L^2((0,T); V)$ to the Cauchy problem (\ref{equ:g_0,g_1}). The solution $\varphi$ satisfies the point-wise estimate (\ref{pointwise}). Hence, $w(\alpha(x,\varphi(x,\tau))) = \alpha(x,\varphi(x,\tau))$ and $\varphi$ is a solution to the Cauchy problem (\ref{finalequation}), as well. 

Finally, from (\ref{monotonicity}) we deduce the $L^\infty((0,T)\times \mathbb{R})$ estimate for the solution $\varphi$ since $\sup_{x\in\mathbb{R}} |\alpha(x,\varphi(x,\tau))| \le M e^{\lambda \tau}$, where $\lambda = \sup_{x\in\mathbb{R}} p(x)$. Furthermore, $\varphi\in L^\infty((0,T)\times \mathbb{R})$, and
\[
\sup_{x\in\mathbb{R}, \tau\in [0,T]} |\varphi(x,\tau)| \le \omega^{-1} 
(M e^{\lambda T} + \max_{x\in\mathbb{R}} |h(x)|).
\]

\end{proof}

\subsection{Application to stochastic dynamic optimal portfolio selection problem}
\label{application}

As an example of the stochastic process (\ref{process_x}), one can consider a portfolio optimization problem with regular cash inflow is an inflow ($\varepsilon>0$)/outflow ($\varepsilon<0$) to a portfolio representing e.g., pension funds savings (c.f. Kilianov\'a and {\v S}ev{\v c}ovi{\v c} \cite{KilianovaSevcovicANZIAM}). In a stylized financial market the stochastic process $\{y_t \}_{t\ge0}$ driven by the stochastic differential equation
\begin{equation}
d y_t = (\varepsilon(y_t) + {\bm{\mu}}^T {\bm{\theta}} y_t)
d t + \sigma(\bm{\theta}) y_t d W_t, \label{processYeps}
\end{equation}
represents a stochastic evolution of the value of a synthetized portfolio $y_t$ consisting of $n$-assets with weights ${\bm{\theta}}=(\theta_1, \cdots, \theta_n)^T$, mean returns ${\bm{\mu}}=(\mu_1, \cdots, \mu_n)^T$, and an $n\times n$  positive definite covariance matrix ${\bm{\Sigma}}$, i.e., $\sigma({\bm{\theta}})^2 ={\bm{\theta}}^T {\bm{\Sigma}}  {\bm{\theta}}$.

We assume that the value $\varepsilon=\varepsilon(y)$ of the inflow/outflow rate also depends on the value $y$ in such a way that  $\varepsilon(y) = 0$ for very small values of $0<y\le y_-$ and $\varepsilon$ is a given constant inflow/outflow rate when the amount of saved money $y$ is sufficiently large, i.e., $y\ge y_+$, where $0<y_-<y_+$ and the function $\varepsilon$ is $C^1$ smooth for all $y>0$. It represents a realistic pension saving model in which there is no inflow/outflow provided that the value $y$ of the portfolio is very small. Based on the logarithmic transformation $x=\ln y$ and It\^{o}'s lemma the stochastic process $\{x_t\}$ satisfies (\ref{process_x}) where 
$\mu(x,{\bm{\theta}}) = {\bm{\mu}}^T{\bm{\theta}} - \frac12 \sigma({\bm{\theta}})^2  +\varepsilon(e^x)  e^{-x}$.

Further generalization of the drift and volatility functions arises from the so-called worst-case portfolio optimization problem investigated by Kilianov\'a and Trnovsk\'a \cite{KilianovaTrnovska}. The volatility function is given by
\[
\sigma({\bm{\theta}})^2 = \max_{{\bm{\Sigma}}\in{\mathcal K}}{\bm{\theta}}^T {\bm{\Sigma}}  {\bm{\theta}}, 
\]
where $\mathcal K$ is a bounded uncertainty convex set of positive definite covariance matrices. In general, only a part of the covariance matrix can be calculated precisely whereas entries are not precisely determined. For instance, if only the diagonal $d$ is known, we have ${\mathcal K}=\{{\bm{\Sigma}}\succ0, \ diag({\bm{\Sigma}})=d\}$. The drift function is given by
\[
\mu(x,{\bm{\theta}}) = \min_{{\bm{\mu}}\in{\mathcal E}}{\bm{\mu}}^T{\bm{\theta}} - \frac12 \sigma({\bm{\theta}})^2  +\varepsilon(e^x) e^{-x} ,
\]
where $\mathcal E$ is a given bounded uncertainty convex set of mean returns.

\begin{remark}

Let us consider a class of utility function characterized by a pair of exponential functions: 
\begin{equation}
u(x) = \begin{cases} - e^{-a_0 x} - c^*, & x \le x^\ast, \\ 
- (a_0/a_1) e^{-a_1 x + (a_1-a_0)x^\ast}, & x > x^\ast, \end{cases} 
\label{eq:utility_DARA}
\end{equation}
where $c^*=e^{-a_0 x^*}(a_0-a_1)/a_1$ and $a_0, a_1\in \mathbb{R}>0$ are given constants. Here, $x^\ast \in \mathbb{R}$ is a point at which the risk aversion changes. Note that $u$ is an increasing $C^1$ function having a jump in the second derivative at the point $x^\ast$. 

If $a_0>a_1>0$, then the utility function $u$ is called DARA (decreasing absolute risk aversion) function. It represents an investor with a non-constant, decreasing risk aversion: the higher the wealth, the lower their risk aversion and hence the higher  exposition of the portfolio to more risky assets. With regard to the paper \cite{Post} by Post, Fang and Kopa, the piece-wise exponential DARA utility functions play an important role in the analysis of decreasing absolute risk aversion stochastic dominance introduced by Vickson \cite{Vickson} (see also \cite{KilianovaSevcovicKybernetika}). Note that the coefficients of absolute risk aversion of the above utility functions $-u^{\prime\prime}(x)/u^{\prime}(x)$ is equal to $a_0$ if $x\le x^*$ or to $a_1$ if $x> x^*$. 

The piece-wise constant function $\varphi_0(x)=-u''(x)/u'(x)$ should be truncated outside of some interval $(-\gamma,\gamma)$, where $\gamma$ is large enough. Then $\varphi_0\in L^2(\mathbb{R})\cap L^\infty(\mathbb{R})$. The underlying utility function is therefore modified by linear functions for $x<-\gamma$ and $x>\gamma$.

Another simple example of a convex-concave utility function is the function $u(x)=\arctan(x)$. Then $\varphi_0(x) = -u''(x)/u'(x) = 2x/(1+x^2)$. Clearly, $\varphi_0\in H=L^2(\mathbb{R})$. It is worth noting that the individual's reduction in marginal utility arising from a loss is absolutely greater than the marginal utility from a financial gain. The utility function is concave (in the domain $x>x^*$), indicating that investors show risk aversion in the domain of gain.  However, investors become risk-seeker when dealing with losses, i.e., the utility function is convex for $x\leq x^*$. 
\end{remark}

\section{Numerical examples}

 First, let us consider a simple example of the decision set $\triangle=\{{\bm{\theta}}\in\mathbb{R}^2,\,  {\bm{\theta}} \ge0, \bm{1}^T {\bm{\theta}} =1\}, n=2, \mu({\bm{\theta}})={\bm{\mu}}^T{\bm{\theta}}, \sigma^2({\bm{\theta}})={\bm{\theta}}^T{\bm{\Sigma}}{\bm{\theta}}$, where ${\bm{\Sigma}}$ is a positive definite covariance matrix and ${\bm{\mu}}$ is a positive vector of mean return. The value function $\alpha=\alpha(\varphi)$ can be explicitly expressed as follows:
\[
\alpha(\varphi)=\left\{ 
\begin{array}{ll}
E^- \varphi + D^-,     & \hbox{if}\ 0<\varphi\le \varphi_*^-,\\
 A - \frac{B}{\varphi} + C\varphi,     & \hbox{if} \ \varphi_*^-<\varphi<\varphi_*^+, \\
E^+ \varphi + D^+,     & \hbox{if}\ \varphi_*^+\le \varphi,
\end{array}
\right.
\]
where $(\varphi_*^-,\varphi_*^+)$ is the maximal interval in which the optimal value $\hat{\bm{\theta}}(\varphi)\in\triangle$ of the function ${\bm{\theta}} \mapsto -{\bm{\mu}}^T{\bm{\theta}} +\frac{\varphi}{2} {\bm{\theta}}^T{\bm{\Sigma}}{\bm{\theta}}$ is strictly positive ($\hat{\bm{\theta}}(\varphi)>0$) for $\varphi\in (\varphi_*^-,\varphi_*^+)$, and $C, E^\pm>0$, $B\ge 0$, $A,D^\pm$ are constants explicitly depending on the covariance matrix ${\bm{\Sigma}}$ and the vector of mean return ${\bm{\mu}}$ such that the function $\alpha$ is $C^1$ continuous at $\varphi_*^\pm$, i.e., $E^\pm=B/(\varphi_*^\pm)^2+C$ and $D^\pm=A-B/\varphi_*^\pm +C\varphi^\pm - E^\pm\varphi_*^\pm$. It is clear that  $\alpha$ is only $C^{1,1}$ continuous function having two points $\varphi_*^\pm$ of discontinuity of the second derivative $\alpha''$. 

If we restrict the decision set to a set consisting of finite number of points, then the value function $\alpha(\varphi)$ is only piece-wise linear. Indeed,
if $\hat\triangle = \{{\bm{\theta}}^1, \cdots, {\bm{\theta}}^k\} \subset \{{\bm{\theta}}\in\mathbb{R}^2,\,  {\bm{\theta}} \ge0, \bm{1}^T {\bm{\theta}} =1\}$ then $\alpha(\varphi)=\min_{i=1,\cdots,k} \alpha^{i}(\varphi)$, where $\alpha^i(\varphi) = E^i \varphi + D^i$ is a linear function with the slope $E^i = (1/2) ({\bm{\theta}}^i)^T{\bm{\Sigma}} {\bm{\theta}}^i >0$ and intercept $D^i = - {\bm{\mu}}^T {\bm{\theta}}^i$.

Figure~\ref{fig:alpha_alphader_alphaderder} a) shows a graph of the value function $\alpha$ corresponding to the Slovak pension fund system. Following the data set from \cite{KMS}, the portfolio comprises of the stocks index with a high mean return $\mu_s=0.1028$ and high volatility $\sigma_s=0.169$ and bonds with mean return $\mu_b=0.0516$ and very low volatility $\sigma_s=0.0082$. Returns on stocks index and bonds have negative correlation $\varrho=-0.1151$. Hence, ${\bm{\mu}}=(\mu_s,\mu_b)^T$ and ${\bm{\Sigma}}_{11}=\sigma_2^2, {\bm{\Sigma}}_{22}=\sigma_b^2, {\bm{\Sigma}}_{12}={\bm{\Sigma}}_{21}=\varrho\sigma_s\sigma_b$. In Figure~\ref{fig:alpha_alphader_alphaderder} a), the  solid blue line corresponds to the convex compact decision set $\triangle=\{{\bm{\theta}}\in\mathbb{R}^2,\,  {\bm{\theta}} \ge0, \bm{1}^T {\bm{\theta}} =1\}$. The piece-wise linear value function $\alpha$ (dotted red line) corresponds to the discrete decision set $\hat\triangle=\{ {\bm{\theta}}^1, {\bm{\theta}}^2, {\bm{\theta}}^3\}\subset \triangle$. It represents the Slovak pension fund system consisting of three funds - the growth fund with ${\bm{\theta}}^1=(0.8,0.2)^T$ (80\% of stocks and 20\% of bonds), the balanced fund with  ${\bm{\theta}}^2=(0.5, 0.5)^T$ (equal proportion of stocks and bonds), and the conservative fund with ${\bm{\theta}}^3=(0, 1)^T$ (only bonds) (c.f. \cite{KMS}). Figure~\ref{fig:alpha_alphader_alphaderder} b) shows the graph of the second derivative $\alpha''_\varphi(\varphi)$ of the value function $\alpha(\varphi)$ corresponding to the convex compact decision set $\triangle$. The first point of discontinuity $\varphi_*^-$ is close to the value 2.
For $n>2$, the number of discontinuities of $\alpha_\varphi^{\prime\prime}$ increases (c.f. \cite{KilianovaSevcovicANZIAM}). In Figure ~\ref{fig:alpha_alphader_alphaderder-dax}, we show another example of the value function and its second derivative for the portfolio consisting of five stocks (BASF, Bayer, Degussa-Huls, FMC, Schering) entering DAX30 German stocks index. The covariance matrix ${\bm{\Sigma}}$ and the vector of yields ${\bm{\mu}}$ is taken from \cite{DDV}. 


\begin{figure}
    \centering
    \includegraphics[width=0.4\textwidth]{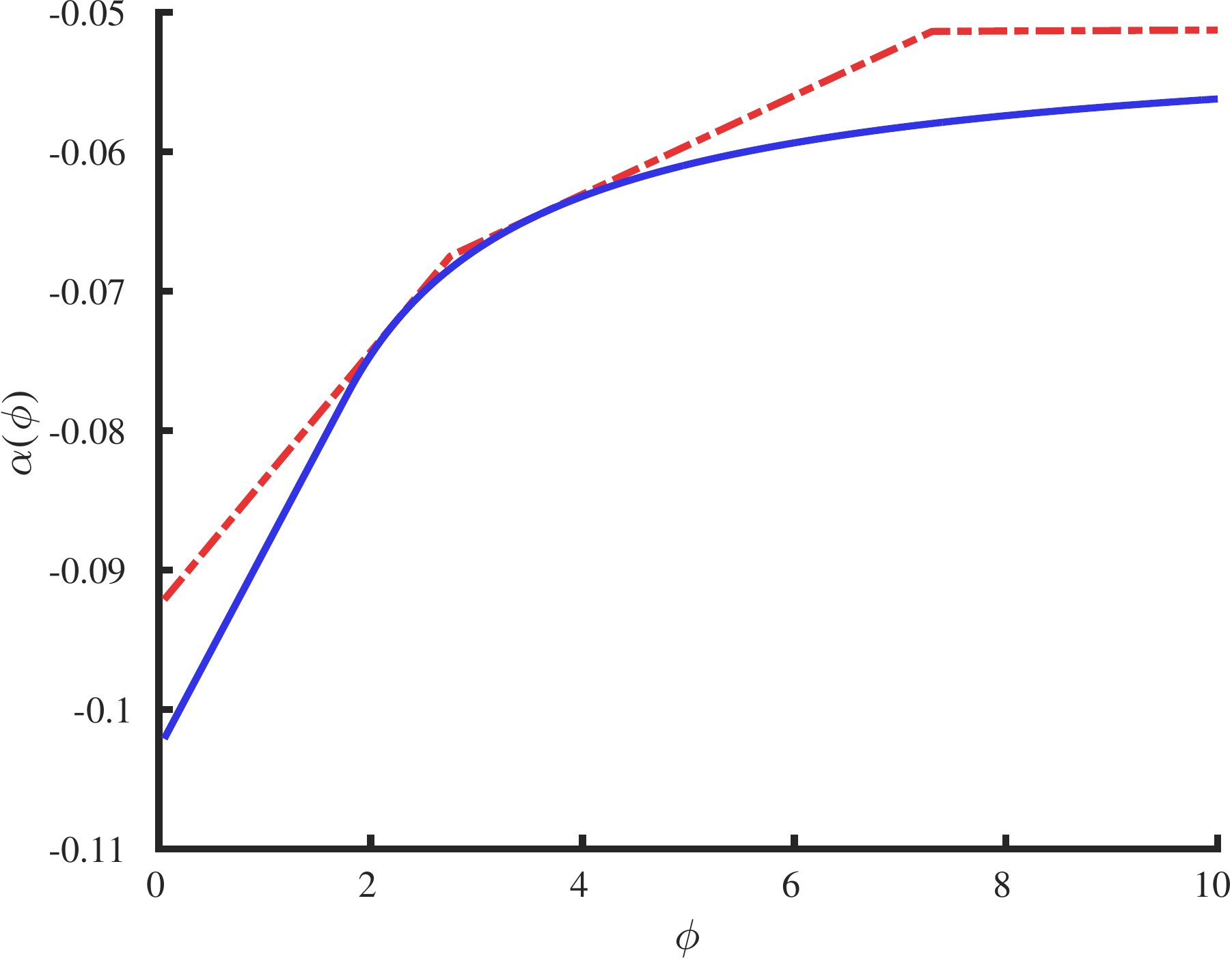}
    \includegraphics[width=0.4\textwidth]{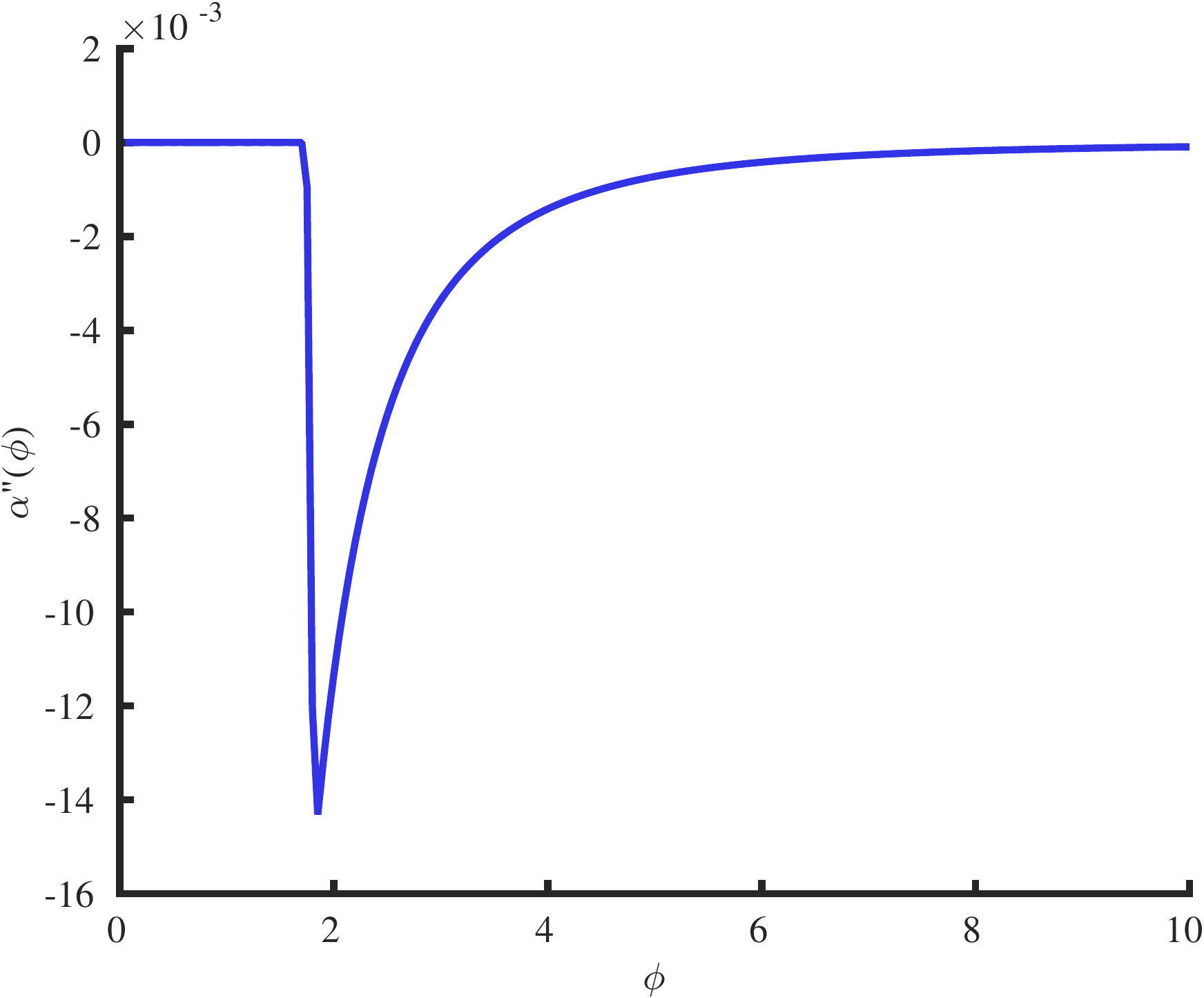}
    
    a) \hglue6truecm \qquad b)
    
    \caption{a) A graph of the value function $\alpha$, b) its second derivative $\alpha''(\varphi)$ for the portfolio consisting of the stocks index and bonds (c.f. \cite{KMS}) for the convex compact decision set $\triangle$. The dotted line in a) corresponds to the discrete decision set $\hat\triangle=\{ {\bm{\theta}}^1, {\bm{\theta}}^2, {\bm{\theta}}^3\}\subset \triangle$.}
    \label{fig:alpha_alphader_alphaderder}
\end{figure}

\begin{figure}
    \centering
    \includegraphics[width=0.4\textwidth]{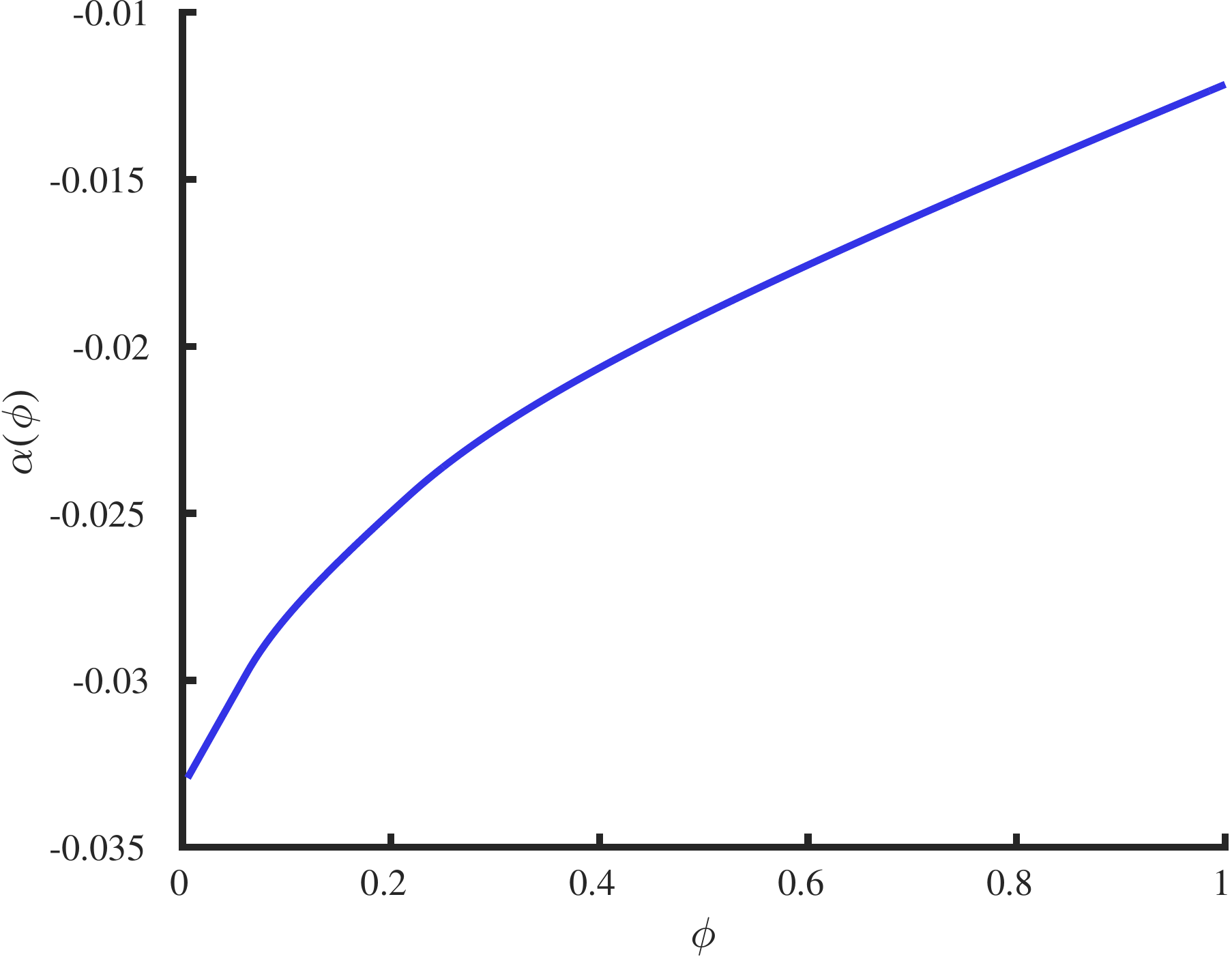}
    \includegraphics[width=0.4\textwidth]{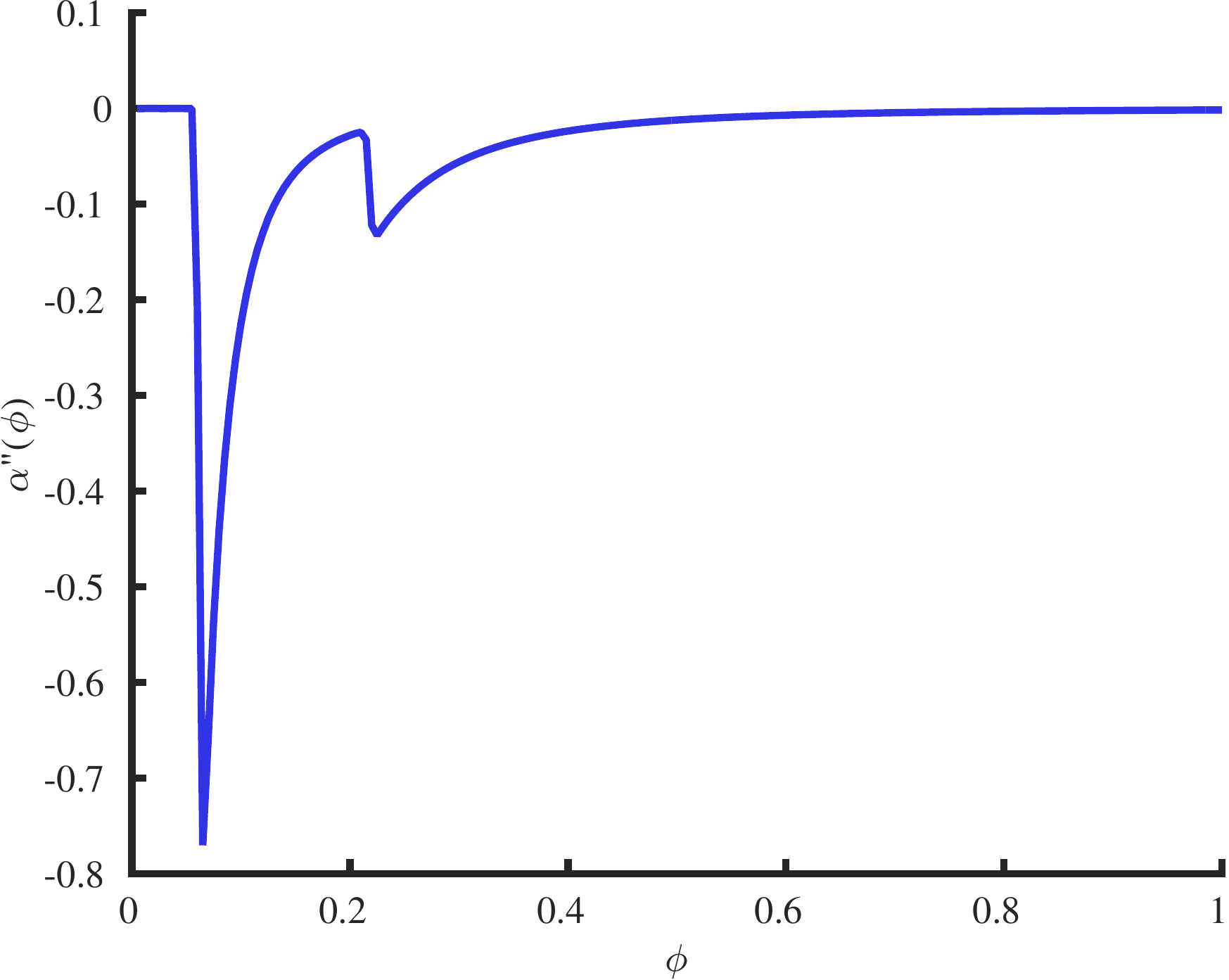}
    
    a) \hglue6truecm \qquad b)
    
    \caption{a) A graph of the value function $\alpha(\varphi)$, and b) the second derivative $\alpha''_\varphi(\varphi)$ corresponding to five stocks from DAX30 index.}
    \label{fig:alpha_alphader_alphaderder-dax}
\end{figure}

The advantage of the Riccati transformation of the original Hamilton-Jacobi-Bellman is twofold. First, the diffusion function $\alpha$ can be computed in advance as result of quadratic optimization problem when the vector $\bm{\mu}$ and the covariance matrix $\bm{\Sigma}$ are given or semidefinite programming problem when they belong to a uncertainity set of returns and covariance matrices (c.f. \cite{KilianovaTrnovska}). Figure \ref{fig:vysledky-DAX} shows the vector of optimal weights $\bm{\theta}$, as a function of the parameter $\varphi$, obtained as the optimal solution to the quadratic optimization problem with the covariance matrix corresponding to the entire  DAX30 index from the year 2017. When the parameter $\varphi$ increases there are more nontrivial weights $\theta_i$. The data set is the same as in the source: Kilianov\'a and \v{S}ev\v{c}ovi\v{c} \cite{KilianovaSevcovicKybernetika}.

\begin{figure}
    \centering
    \includegraphics[width=0.45\textwidth]{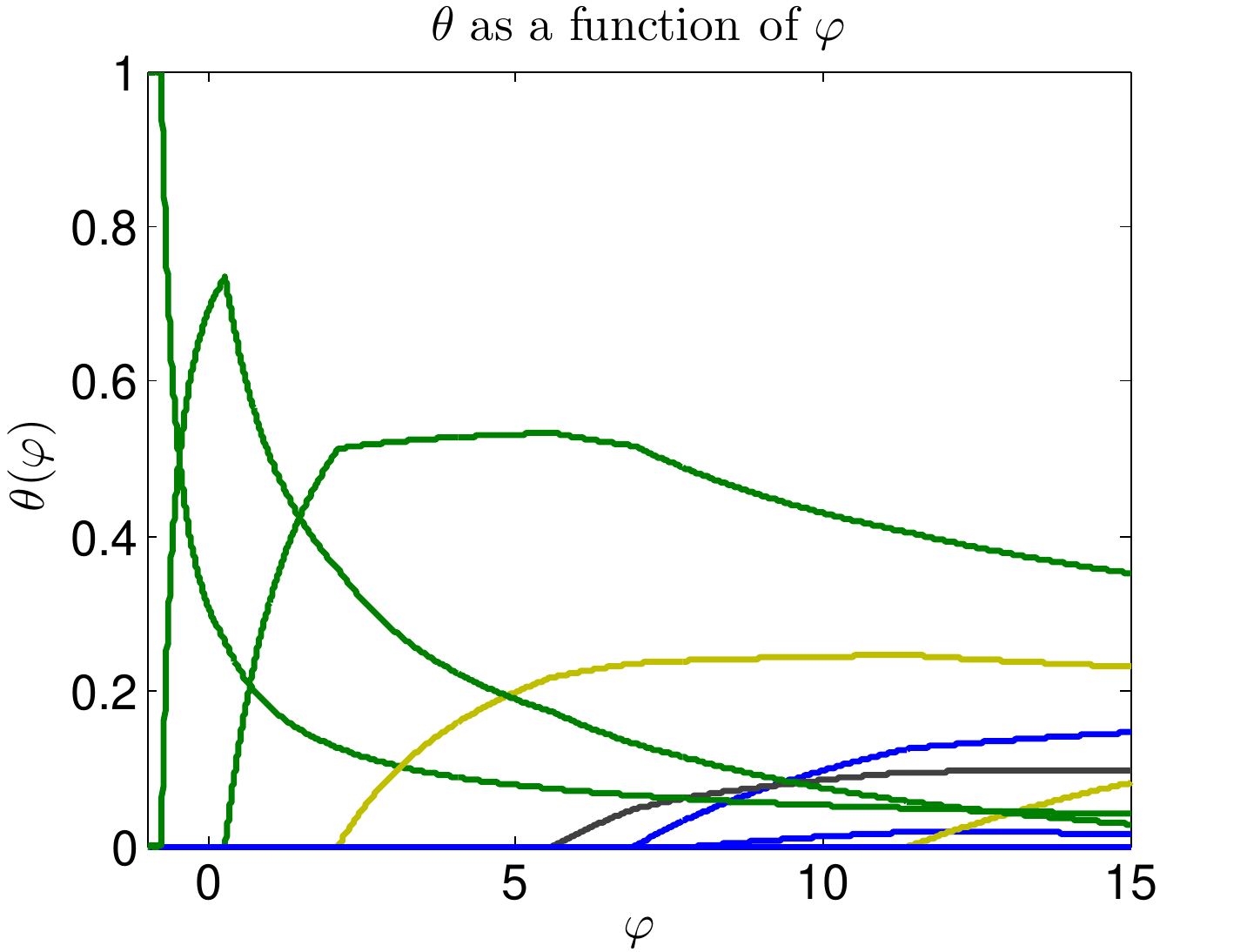} \\
    
    \caption{The optimal vector $\bm{\theta}=(\theta_1, \cdots,\theta_n)^T$ as a function of $\varphi$ for the German DAX30 index. }
    \label{fig:vysledky-DAX}
\end{figure}

Second, in contrast to the fully nonlinear character of the original Hamilton-Jacobi-Bellman equation (\ref{eq_HJB}), the transformed governing equation (\ref{eq_PDEphi-cons}) represents a quasi-linear parabolic equation in the divergence form. Hence efficient numerical schemes can be constructed for this class of equation. In our computational experiments, we follow the finite volume discretization scheme proposed and investigated by Kilianov\'a and \v{S}ev\v{c}ovi\v{c} \cite{KilianovaSevcovicKybernetika, KilianovaSevcovicANZIAM, KilianovaSevcovicJIAM}). In Figure~\ref{fig:vysledky1} a), we present results of time dependent sequence of profiles $\varphi(x,\tau)$ for a constant initial condition $\varphi_0\equiv 9$. In Figure~\ref{fig:vysledky1} b) we show profiles of solutions for the initial condition $\varphi_0$ attaining four decreasing values $\{9,8,7,6\}$. It represents DARA (decreasing absolute risk aversion) utility function. Clearly, these profiles are lower than those with constant $\varphi_0\equiv 9$. Therefore, the optimal vector $\bm{\theta}(x,\tau)$ contains more risky assets at any $x$ and time $\tau$ (see Figure~\ref{fig:vysledky-DAX}).  

Figure~\ref{fig:vysledky-DAX} also shows that there are only a few relevant assets out of the set of 30 assets
entering the DAX30 index. The figure also reveals the highest portion of Merck stocks (the first decreasing line in the plot) starting from $100\%$ representation in the optimal portfolio for very small values of $\varphi=\varphi(x,\tau)$. It corresponds to the early period of saving $\tau\approx 0$ and low account values of $x$. Although with the highest volatility, it is indeed reasonable to invest in an asset with the highest expected return when the account value $x$ is low, in early times of saving. We can also observe a fast decrement of the Merck weight for increasing risk averesion value $\varphi$. On the other hand, the Fresenius Medical (see Figure~\ref{fig:vysledky-DAX}, the yellow line) has the lowest volatility out of the considered five assets displayed in , and third-best mean return, which is reflected in its major representation in the portfolio for higher values $\varphi=\varphi(x,\tau)$.

\begin{figure}
\centering
\includegraphics[width=0.45\textwidth]{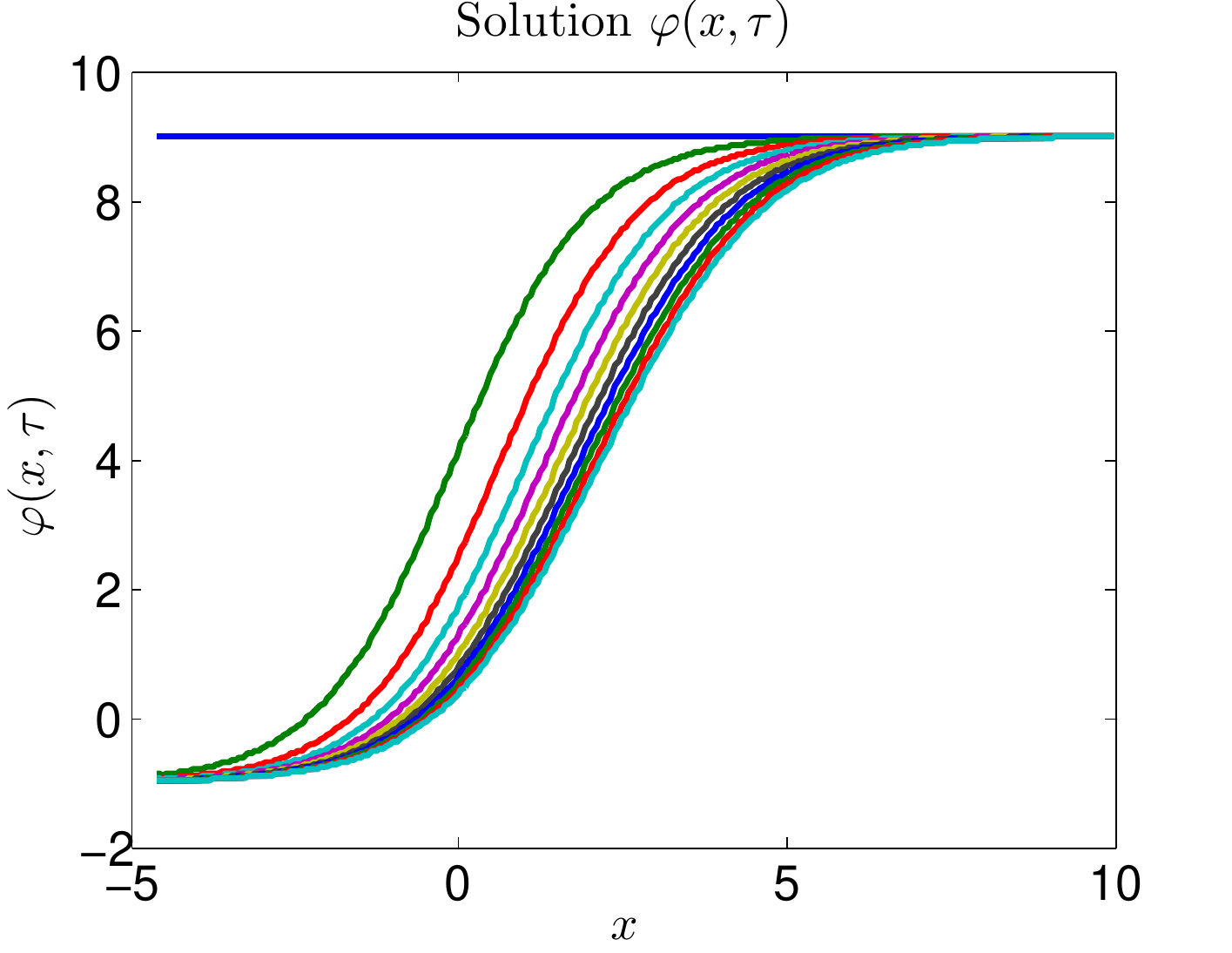} 
\includegraphics[width=0.45\textwidth]{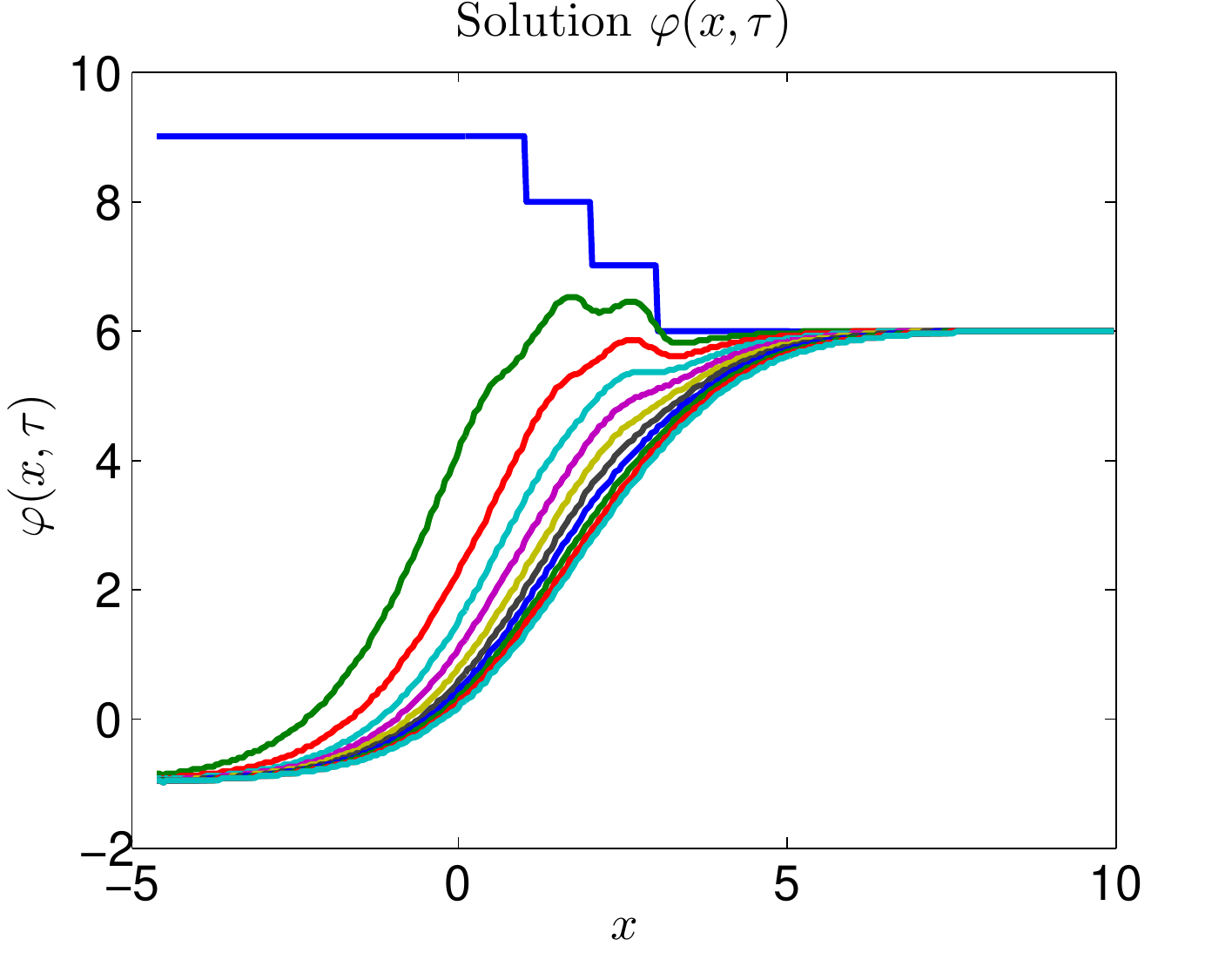}     

    \caption{Solutions $\varphi(x,\tau)$ for the utility function $u$ with constant $a_0=a_1=9$ (left) and for the DARA utility function with $a_0=9$, $a_1=6$, $x^\ast=2$ (right).}
    \label{fig:vysledky1}
\end{figure}

\section{Conclusions}
In this paper, we investigated and analyzed the existence and uniqueness of a solution to the Cauchy problem for the parabolic PDE (\ref{generalPDE}) in a suitable Sobolev space using monotone operator approach. We utilized the Banach's fixed point theorem and Fourier transform technique to prove the existence result in an abstract setting. As a financial application in one-dimensional space, we considered a fully nonlinear evolutionary Hamilton-Jacobi-Bellman (HJB) parabolic equation arising from portfolio optimization selection, where the goal is to maximize the conditional expected value of the terminal utility of the portfolio. Using the so-called Riccati method for transformation, the fully nonlinear HJB equation is transformed into a quasilinear parabolic equation, which is equivalent to the proposed result after some shift in the operator. Under some assumptions, we obtained that the diffusion function to the quasilinear parabolic equation is globally Lipschitz continuous, which is a crucial requirement for solving the Cauchy problem. Some numerical examples of the proposed results were presented.

\section*{Acknowledgements}
The authors were supported by VEGA 1/0062/18 (D\v{S}) and DAAD-MS MATTHIAS-2020 (CU) grants.


\begin{thebibliography}{10}
\newcommand{\enquote}[1]{`#1'}

\bibitem{AI}
R.~Abe and N.~Ishimura.
\newblock \enquote{{Existence of solutions for the nonlinear partial
  differential equation arising in the optimal investment problem.}}
\newblock {\em Proc. Japan Acad., Ser. A\/} {\bf 84}~(1) (2008), 11--14.





\bibitem{Bertsekas}
D.~P. Bertsekas.
\newblock {\em {Dynamic programming and stochastic control.}\/} ({Academic
  Press, New York}, 1976).




\bibitem{DDV}
G.~Deelstra, I.~Diallo, and M.~Vanmaele.
\newblock\enquote{Bounds for Asian basket options}.
\newblock{\em Journal of Computational and Applied Mathematics} {\bf 218} (2008), 215-228.

\bibitem{Federico}
S.~Federico, P.~Gassiat, and F.~Gozzi.
\newblock \enquote{{Utility maximization with current utility on the wealth:regularity of solutions to the HJB equation.}}
\newblock {\em Finance Stoch\/} {\bf 19} (2015), 415--448.







\bibitem{IshSev}
N.~Ishimura, and D.~\v{S}ev\v{c}ovi\v{c}.
\newblock \enquote{{On traveling wave solutions to a Hamilton-Jacobi-Bellman
  equation with inequality constraints.}}
\newblock {\em Japan J. Ind. Appl. Math.\/} {\bf 30}~(1) (2013), 51--67.




\bibitem{KilianovaSevcovicANZIAM}
S.~Kilianov\'a, and D. \v{S}ev\v{c}ovi\v{c}.
\newblock \enquote{A Transformation Method for Solving the Hamilton-Jacobi-Bellman Equation for a Constrained Dynamic Stochastic Optimal Allocation Problem}.
\newblock {\em ANZIAM Journal} {\bf 55} (2013), 14--38.

\bibitem{KilianovaSevcovicKybernetika}
S.~Kilianov\'a, and D. \v{S}ev\v{c}ovi\v{c}.
\newblock \enquote{Expected Utility Maximization and Conditional Value-at-Risk Deviation-based 
Sharpe Ratio in Dynamic Stochastic Portfolio Optimization}.
\newblock {\em Kybernetika} 54(6) (2018), 1167-1183.



\bibitem{KilianovaSevcovicJIAM}
S.~Kilianov\'a, and D. \v{S}ev\v{c}ovi\v{c}.
\newblock \enquote{Dynamic intertemporal utility optimization by means of Riccati transformation of Hamilton-Jacobi Bellman equation}.
\newblock{\em Japan Journal of Industrial and Applied Mathematics,} 36(2) (2019), 497-517.



\bibitem{KMS}
S.~Kilianov\'a, I. Melicher\v{c}\'\i k, D. \v{S}ev\v{c}ovi\v{c}.
\newblock \enquote{Dynamic Accumulation Model for the Second Pillar of the Slovak Pension System}, 
\newblock{Finance a uver - Czech Journal of Economics and Finance}, {\bf 56}(11-12) (2006), 506-521.


\bibitem{KilianovaTrnovska}
S.~Kilianov\'a, and M. Trnovsk\'a.
\newblock \enquote{Robust Portfolio Optimization via solution to the Hamilton-Jacobi-Bellman Equation}. 
{\em Int. Journal of Computer Mathematics} {\bf 93} (2016), 725--734. 


\bibitem{Klatte}
D.~Klatte.
\newblock \enquote{On the {L}ipschitz behavior of optimal solutions in
  parametric problems of quadratic optimization and linear complementarity}.
\newblock {\em Optimization: A Journal of Mathematical Programming and
  Operations Research\/} {\bf 16}~(6) (1985), 819--831.







  




  

\bibitem{MS}
Z.~Macov{\'a}, and D.~{\v{S}}ev{\v{c}}ovi{\v{c}}.
\newblock \enquote{Weakly nonlinear analysis of the
  {H}amilton-{J}acobi-{B}ellman equation arising from pension savings
  management}.
\newblock {\em Int. J. Numer. Anal. Model.\/} {\bf 7}~(4) (2010), 619--638.




\bibitem{Meyer}
J. C. Meyer, and D. J. Needham. 
\newblock \enquote{Extended weak maximum principles for parabolic partial differential inequalities on unbounded domains}. 
\newblock {\em Proc. R. Soc. Lond. Ser. A Math. Phys. Eng. Sci.} {\bf 470} (2014), 20140079.



\bibitem{milgrom_segal2002}
P.~Milgrom, and I.~Segal.
\newblock \enquote{Envelope theorems for arbitrary choice sets}.
\newblock {\em Econometrica\/} {\bf 70}~(2) (2002), 583--601.


\bibitem{Post}
T.~Post, Y.~Fang, and M.~Kopa.  
\newblock \enquote{Linear Tests for DARA Stochastic Dominance.\/}
\newblock{\em Management Science\/}
{\bf 61} (2015), 1615--1629. 




\bibitem{Protter}
M. H. Protter, and H. F. Weinberger.
\newblock {\em Maximum principles in differential equations}, \newblock{Springer Science \& Business Media, 2012.}





\bibitem{SSM}
D.~\v{S}ev\v{c}ovi\v{c}, B.~Stehl\'{\i}kov\'a and K.~Mikula.
\newblock {\em Analytical and numerical methods for pricing financial
  derivatives.\/} (Nova Science Publishers, Inc., Hauppauge, 2011).
 
\bibitem{Vickson} 
R.~G.~Vickson.
\newblock \enquote{Stochastic Dominance for Decreasing Absolute Risk Aversion.} 
\newblock{\em Journal of Financial and Quantitative Analysis\/} 
{\bf 10}, (1975), 799--811.

 


\bibitem{Showalter}
R.~E. Showalter
\newblock \enquote{{Monotone operators in Banach space and nonlinear partial differential equations.}}
\newblock {\em American Mathematical Soc.\/} {\bf 49}, (2013)
\bibitem{Barbu}
V. Barbu
\newblock \enquote{{Nonlinear differential equations of monotone types in Banach spaces.}}
\newblock {\em Springer Science \& Business Media\/} (2010).

\bibitem{Wu et al}
Wu et al.
\newblock \enquote{{Blow-up of solutions for a semilieanr parabolic equation involving variable source and positive initial energy.}}
\newblock {\em Applied Mathematics Letters.\/} {\bf 26}~(5) (2013), 539-543.

\bibitem{PaoRuan}
C. V. Pao and W.H. Ruan
\newblock \enquote{{Positive solutions of quasilieanr parabolic systems with Dirichlet boundary condition.}}
\newblock {\em Journal of Differential Equation.\/} {\bf 248}~(5) (2010), 1175--1211.

\end{thebibliography}
\end{document}